%% file: ArXiv.tex
\documentclass[12pt]{article}
\usepackage{amsfonts,amsmath,amssymb,amsthm,graphicx}
\usepackage{fullpage}
\usepackage{changepage}
\usepackage{enumerate}
\usepackage{algorithm}
\usepackage{algorithmic}
\usepackage{scalerel}
\usepackage{accents}
\usepackage{hyperref}
\usepackage{xfrac}
\usepackage[margin=1in]{geometry}
\usepackage{setspace}
\usepackage{thm-restate}

\usepackage{subfig}
\usepackage{natbib}
\usepackage{csquotes}
\usepackage{comment}
\setcitestyle{authoryear,open={(},close={)}}
\usepackage[capitalize]{cleveref}

\usepackage{tikz}
\usetikzlibrary{patterns}
\usetikzlibrary{intersections}
\usepackage{pgfplots}

\pgfplotsset{compat = newest}

\theoremstyle{plain}
\newtheorem{theorem}{Theorem}
\newtheorem{lemma}{Lemma}

\theoremstyle{plain}
\newtheorem{definition}{Definition}
\newtheorem{example}{Example}
\newtheorem{assumption}{Assumption}

\theoremstyle{plain}

\include{notation}

\include{math}

\let\oldparagraph\paragraph
\renewcommand{\paragraph}[1]{\oldparagraph{#1.}}

\providecommand{\keywords}[1]
{
  \small	
  \textbf{\textit{Keywords---}} #1
}

\begin{document}


\title{Selling Data to an Agent with Endogenous Information\thanks{The author thanks NSF CCF AitF-1733860 and NSF SES-1947021 for financial support. 
The two-page abstract of this paper appeared in the Proceedings of the 23rd ACM Conference on Economics and Computation (EC'22).
The author thanks Yiding Feng, Yingni Guo, Jason Hartline, Rad
Niazadeh, Boli Xu and participants of the Stony Brook International Conference on Game Theory for helpful comments and suggestions.}}
\author{Yingkai Li\thanks{Cowles Foundation for Research in Economics, Yale University.
Email: \texttt{yingkai.li@yale.edu}}}
\date{}

\maketitle

\begin{abstract}
We consider a model of a data broker selling information to a single agent to maximize his revenue. The agent has a private valuation of the additional information, and upon receiving the signal from the data broker, the agent can conduct her own experiment to refine her posterior belief on the states with additional costs. To maximize expected revenue, only offering full information in general is suboptimal, and the optimal mechanism may contain a continuum of menu options with partial information to prevent the agent from having incentives to acquire additional information from other sources. However, our main result shows that the additional benefit from price discrimination is limited, i.e., posting a deterministic price for revealing full information obtains at least half of the optimal revenue for arbitrary prior and cost functions. 
\end{abstract}


\section{Introduction}\label{sec:intro}

Information can help decision makers refine their knowledge and make better decisions 
when there is uncertainty regarding their environments. 
The desire for additional information creates opportunities for data brokers to collect payments from decision makers. 
There is a large market for selling information 
such as selling consumer data (e.g., Acxiom and Datalogix),
user profiles (e.g., Facebook),
credit reports (e.g., Experian, TransUnion),
or cookies from web usage \citep*[c.f.,][]{bergemann2015selling}.
Data brokers extract considerable revenue from providing valuable information to decision makers.

In this paper, we consider the problem of maximizing the revenue of the data broker, 
where the agent can endogenously acquire additional information. 
Specifically, there is an unknown state, and 
both the data broker and the agent have a common prior over the set of possible states.
The data broker can offer a menu of experiments for revealing the states with associated prices to the agent. 
Then, the agent selects the expected utility maximizing entry from the menu 
and pays the corresponding price to the data broker.
The agent has a private valuation of information
and can acquire additional costly information upon receiving the signal from the data broker. 
The literature has acknowledged the possibility that agents can conduct their own experiments to be privately informed of states \citep*[e.g.,][]{bergemann2018design}.
The distinct feature in our model is that the decision of whether to acquire additional information is endogenous. 
Specifically, after receiving the signal from the data broker, 
the agent can subsequently acquire additional information at a cost. 
For example, the agent is a decision maker who makes investment decisions based on the current state of the financial market. 
The agent will first acquire information from the data broker (for example, a consulting firm), 
and based on her posterior, 
she can potentially repeatedly collect additional information from
other sources (e.g., social media) before making her investment decision.

Another example captured by our model is where the agent is a firm that sells products to consumers, 
and the information the data broker provides is a market segmentation of consumers. 
The firm has a private and convex cost for producing different levels
of quality of its product, 
and the firm can conduct experiments (e.g., sending surveys to potential consumers) with additional costs
to further segment the market after receiving the information from the data broker.\footnote{\citet*{yang2020selling} studies a similar model where the firm cannot conduct its own market research to refine its knowledge about consumers
and has linear costs for producing the goods.} 
In addition, in our model, 
we allow the firm to repeat the market research until it is not beneficial to do so, 
i.e., when the cost of information exceeds the marginal benefits of information. 
This captures the situation where 
the firm can decide the date for announcing a product to the market, 
and before the announcement, 
the firm sends out surveys to potential consumers each day to learn the segmentation of the market. 
At the end of each day, 
the firm receives an informative signal through the survey
and decides whether to continue the survey the next day
or stop the survey and announce the product with corresponding market prices to the public.

We fully characterize the revenue-optimal mechanism 
under the assumption of linear valuation, 
i.e., the value of the agent for any posterior distribution
is simply the product of her private type 
and the value of the posterior distribution.\footnote{We briefly
  discuss the general properties of the optimal mechanisms without the linearity assumption in \cref{apx:general}.} 
In the examples we provided in previous paragraphs, 
both the decision maker who chooses an optimal action to maximize her payoff based on the posterior belief 
and the firm that sells products to consumers to maximize its revenue 
satisfy the linear valuation assumption.
Essentially, this condition assumes that the private type of the agent represents her value of additional information, 
and there is a linear structure of this preference. 
It excludes the situation where the private type of the agent represents an exogenous private signal correlated with the state \citep*[e.g.,][]{bergemann2018design}.

With linear valuations, 
when the agent can acquire additional costly information,
there exists a threshold type $\monopt$ such that 
(1) for any type $\type\geq \monopt$, 
the optimal mechanism reveals full information to the agent,
and (2) for any type $\type < \monopt$, the optimal mechanism may reveal partial information
and the individual rationality constraint always binds. 
The first statement is the standard no distortion at the top observation in the optimal mechanisms.  
The second statement suggests that the optimal mechanism may
discriminate against lower types of the agent 
by offering the option of revealing partial information to the agent at lower prices. 
Moreover, the allocations and prices for those lower types are set such that the agent is exactly indifferent between participation and choosing the outside option (by conducting her own experiments at additional costs). 
This is a clear distinction from selling information without endogenous information 
where the individual rationality constraint for the lowest type is sufficient for characterizing the optimal mechanism in the latter case. 

Our characterizations suggest that the optimal mechanisms for selling information may be complex and contain a continuum of menu options. 
Moreover, it would require the seller to have the ability to commit to using randomization for revealing partial information.
A criticism of such complex mechanisms is their applicability in practice compared to simple but potentially suboptimal mechanisms. 
There is a long line of work providing theoretical justifications for simple mechanisms in practice
by demonstrating the approximate optimality of simple mechanisms when the optimal mechanisms are complex, 
e.g., \citet*{hartline2009simple,har-12,babaioff2020simple}.
In the setting of selling information, 
a natural candidate for simple mechanisms is posting a deterministic price for revealing full information. 
Under the assumption of linear valuation, 
we show that when the prior distribution over states is sufficiently informative,
posted pricing for revealing full information
is the revenue-optimal mechanism.
Moreover, without any restriction on the prior distribution or the cost function,
posted pricing for revealing full information is suboptimal but achieves at least half of the optimal revenue in the worst case.
This approximation result indicates that price discrimination and the ability to provide partial information are not salient features for data selling mechanisms to approximate the optimal revenue.

Finally, we apply our characterization of the optimal mechanism to several economic applications, e.g.,
when the data buyer is a decision maker minimizing the loss of making predictions
or a firm selling products to consumers to maximize profits. 
Our numerical results show that in those applications, the revenue
loss from pricing for full information could be significantly smaller
than our worst-case bound. 
This provides strong justification for adopting simple pricing mechanisms in practice.

\subsection{Related Work}\label{sec:related}

There is a large literature on selling information to agents with uncertainty over states. 
Those papers can be classified into two categories according to the agents' private types. 
The first category is when the agents' private types
represent their willingness to pay for different experiments \citep*[e.g.,][]{yang2020selling,smolin2020disclosure,liu2021optimal}. 
In this case, the private types of the agents are assumed to be independent of the realization of the state, 
and hence the private types do not affect the belief updating process for receiving the signals. 
The second category is when the agents' private types
represent their private signals that are informative about the states \citep*[e.g.,][]{admati1986monopolistic,admati1990direct,bergemann2018design,bergemann2021is}. 
In this case, agents have heterogeneous prior beliefs about the states, 
and they will update their posteriors accordingly upon receiving the signals from the seller. 
Note that in both lines of work, the private types of the agents are given exogenously, 
and hence it is a pure adverse selection model. 
Our paper assumes that the private types of the agents
represent their preferences for different experiments,
which are assumed to be independent of the realization of the states. 
The distinct feature is that we allow agents to endogenously acquire costly signals that are informative about the unknown states.
Thus, the main focus of our paper is the interaction between adverse selection and moral hazard, 
and we will provide characterizations of the optimal mechanisms in this setting.

Our paper is relevant to the literature on mechanism design with endogenous information.
\citet*{cremer1992gathering} considers a model of endogenous information in a contract design model. 
The main distinction between their model and our paper is the timeline of the agent. 
In their paper, the agent gathers information before signing the contract, 
while in our model, the agent can observe additional information after the interaction with the data broker. 
This difference in timeline also distinguishes our model from the literature on auction with endogenous entry \citep*{menezes2000auctions}
and auction with costly information acquisition for buyers \citep*{bergemann2002information,shi2012optimal,li2019efficient,mensch2022screening}, 
where those papers assume that the agents make the information acquisition decision before interacting with the seller,
and the mechanism offered by the seller distorts the agents' incentives to learn their valuation.
\citet*{argenziano2016strategic} analyzed a strategic communication setting where there are misaligned interests between the decision maker and a biased expert. 
They show that the decision maker may prefer a decision based on the advice of the expert even if the cost of information is lower for the decision maker. 
In the persuasion setting, \citet*{matyskova2023bayesian} show that the ability to acquire additional information limits the sender's ability to persuade 
and, perhaps surprisingly, a lower cost of information could be
undesirable for the receiver since it may induce a less informative signal from the sender. 
\citet*{bloedel2020persuading,wei2021persuasion} also consider a persuasion setting where the receiver's information is endogenous. 
However, in their model, the endogeneity arises because it is costly for the receiver to process the information transmitted from the sender.

The model in this paper also contributes to the broader domain of information design and Bayesian persuasion \citep*[e.g.,][]{rayo2010optimal,kamenica2011bayesian}, 
particularly the disclosure of information in auctions. 
\citet*{esHo2007optimal} and \citet*{li2017discriminatory} consider the setting of selling information to the buyer before the auction 
to maximize revenue, 
and \citet*{wei2020reverse} consider the variant where the buyer can walk away without paying the seller 
after receiving the information.
\citet*{bergemann2015selling} models the problem of a data provider who sells cookies to match firms with customers, 
and \citet*{bergemann2015limits} explores the set of possible outcomes
that can be implemented by in a monopoly model by providing the segmentations of the customers.
For a detailed discussion on information in markets, see the survey by \citet*{bergemann2019markets}.

\section{Model}\label{sec:model}
There is a single agent making decisions facing uncertainly over the state space $\states\subseteq\reals$.
Let~$\prior$ be the prior distribution over the states. 
The agent has a private type $\type \in \types=[\underline{\type},\bar{\type}]\subseteq\reals_+$
drawn from a commonly known distribution $\dist$ with differentiable density function $\density$ that is strictly positive everywhere in the support.
The expected utility of the agent given the posterior belief $\posterior \in \Delta(\states)$ is 
$\V(\posterior,\type)$ when her type is $\type$.
We assume that $\V$ is convex in $\posterior$ for any type~$\type$.\footnote{The convexity ensures that for any type $\type$, the agent has higher value for Blackwell more informative experiments. 
This condition is satisfied whenever the agent is a decision maker who maximizes his expected utility. }
There is a data broker who attempts to sell information to the agent to maximize his profit
by committing to an experiment that signals the state. 
Note that an experiment is a mapping $\signal:\Omega \to \Delta(\signalspace)$, 
where $\signalspace$ is the signal space.
Let $\signals$ be the set of all possible experiments.

In this paper, upon receiving a signal $\s\in \signalspace$,
the agent can conduct her own experiment to further refine her posterior belief on the state
for additional costs. 
Let $\experiments\subseteq\signals$ be the set of possible experiments that can be conducted by the agent. 
The cost of experiment $\signal$ given posterior belief $\posterior$ of the agent
is denoted by $\costagent(\signal, \posterior)\geq 0$.
Let $\fullexp$ be the experiment that reveals full information, i.e., $\fullexp(\state)=\state$ for any $\state\in\states$, 
and let $\noexp$ be the null experiment that reveals no information with zero cost. 
In this paper, we assume that $\noexp\in\experiments$,\footnote{ Intuitively, this assumes that the agent can always choose not to conduct any additional experiment 
and pays no extra cost. }
$\costagent(\noexp, \posterior) > 0$, 
and $\costagent(\signal, \posterior) > 0$
for any $\signal\neq\noexp$.

A mechanism of the data broker is a menu of experiments and associated prices $\{(\signal_i,\price_i)\}$.\footnote{It is without loss of generality to restrict attention to offering a deterministic experiment in each menu entry
since any distribution over experiments can be equivalently
represented by a deterministic experiment with an augmented signal space.}
The timeline of the model is illustrated as follows. 
\begin{enumerate}
\item The data broker commits to a mechanism $\mech=\{(\signal_i,\price_i)\}$.
\item The agent chooses entry $(\signal,\price) \in \mech$,
and pays price $\price$ to the data broker.
Then, state $\state\in \states$ is realized 
according to prior $\prior$,
and the data broker sends signal $\s \sim \signal(\state)$ to the agent. 
\item Upon receiving signal $\s$, 
the agent forms posterior belief $\posterior$, 
chooses an experiment $\experiment\in\experiments$, 
and pays cost $\costagent(\experiment,\posterior)$. 
\item The agent receives a signal $\hat{\sa} \sim \experiment(\state)$, 
forms refined posterior belief $\hat{\posterior}$, 
and receives expected reward $\V(\hat{\posterior},\type)$.
\end{enumerate}

By the revelation principle, it is without loss to consider the revelation mechanism, 
that is, the data broker commits to a mapping from types to experiments $\alloc:\types \to \signals$
and the expected payment rule $\price: \types \to \reals$.
By slightly overloading the notation, 
denote 
\begin{align*}
\V(\posterior, \experiments, \type) \triangleq 
\max_{\experiment\in\experiments} \,
\expect[\hat{\posterior} \sim \experiment\given\posterior]{\V(\hat{\posterior}, \type)}
- \costagent(\experiment,\posterior)
\end{align*}
as the maximum utility of the agent given belief $\posterior$, 
the set of possible experiments $\experiments$,
and private type $\type$.
Here, the notation $\experiment\given\posterior$ represents the distribution over posteriors induced by experiment $\experiment$ given prior belief $\posterior$.

\begin{definition}\label{def:ic ir}
Mechanism $\mech=(\alloc, \price)$ is \emph{incentive compatible} if 
for any type $\type, \type'\in \types$, we have 
\begin{align*}
\expect[\posterior\sim\alloc(\type)\given\prior]{
\V(\posterior, \experiments, \type)} - \price(\type)
\geq \expect[\posterior\sim\alloc(\type')\given\prior]{
\V(\posterior, \experiments, \type)} - \price(\type'),
\end{align*}
and mechanism $(\alloc, \price)$ is \emph{individually rational} if for any type $\type \in \types$, we have 
\begin{align*}
\expect[\posterior\sim\alloc(\type)\given\prior]{
\V(\posterior, \experiments, \type)} - \price(\type)
\geq \V(\prior, \experiments, \type).
\end{align*}
\end{definition}
In this paper, without loss of generality, 
we focus on mechanisms $(\alloc, \price)$ that are incentive compatible 
and individually rational.
The goal of the data broker is to maximize expected revenue
$\Rev(\mech) \triangleq \expect[\type\sim \dist]{\price(\type)}$.

For any experiment $\experiment \in \experiments$ 
and any mapping $\kappa: \signalspacea \to \experiments$, 
let $\kappa \circ \experiment$
represent the experiment such that the agent first conducts experiment $\experiment$, 
and conditional on receiving the signal $\sa\in\signalspacea$, 
the agent continues with experiment $\kappa(\sa)$ to further refine her posterior belief. 
For any belief~$\posterior$, let $\hat{\posterior}_{\experiment, \sa, \posterior}$ be the posterior belief of the agent when she conducts experiment~$\experiment$
and receives signal $\sa$. 
Throughout this paper, we make the sequential learning proof assumption from \citep{bloedel2020cost} on the set of possible experiments and the cost function. 

\begin{assumption}[\citealp{bloedel2020cost}]\label{asp:concatenation}
For any experiment $\experiment \in \experiments$
and any mapping $\kappa: \signalspacea \to \experiments$, 
we have $\kappa \circ \experiment \in \experiments$. 
Moreover, for any belief $\posterior$, 
we have 
\begin{align*}
\costagent(\kappa \circ \experiment, \posterior) \leq
\costagent(\experiment, \posterior)
+ \int_{\states} \int_{\signalspacea} 
\costagent(\kappa(\sa), \hat{\posterior}_{\experiment, \sa,\posterior})
\dd \experiment(\sa\given \state) \dd \posterior(\state).
\end{align*}
\end{assumption}

Intuitively, \cref{asp:concatenation} assumes that the set of possible experiments is closed under sequential learning, 
and the cost function exhibits a preference for one-shot learning.\footnote{\citet*{bloedel2020cost} provide a characterization for the cost function to be indifferent for one-shot learning with additional regularity assumptions. 
In our paper, we only need to assume a weak preference for one-shot learning, 
and the additional regularity assumptions in \citet*{bloedel2020cost} are not needed here.
}
This captures the scenario where the agent can repeatedly conduct feasible experiments based on her current posterior belief. 
Next, we illustrate several examples that satisfy the above assumptions. 
\begin{itemize}
\item $\experiments$ is a singleton. 
In this case, the unique experiment $\noexp \in \experiments$ is null experiment with zero cost.

\item $\experiments$ is the set of all possible experiments, i.e., $\experiments = \signals$. 
The cost function $\costagent$ is the reduction in uncertainty, 
i.e., 
$\costagent(\experiment,\posterior) = \entropy(\posterior) - \expect[\hat{\posterior}\sim \experiment\given\posterior]{\entropy(\hat{\posterior})}$
where $\entropy$ is a concave function that measures the uncertainty. 
Possible choices of function $\entropy$ include the entropy function \citep*[e.g.,][]{sims2003implications}
or more generally the uniformly posterior separable cost functions \citep*[e.g.,][]{caplin2022rationally}.
It is easy to verify that uniformly posterior separable cost functions satisfy \cref{asp:concatenation}.

\item $\experiments$ is the set of experiments generated by $\noexp$ and $\experiment'$ through sequential learning, 
where $\noexp$ is experiment that reveals no additional information with zero cost,
and $\experiment'$ is an informative experiment that signals the state with fixed cost, 
i.e., there exists a constant $c>0$
such that $\costagent(\experiment', \posterior) = c$
for all posteriors $\posterior$.
In this case, the agent can choose experiment $\experiment'$ as long as it is beneficial for her given her current belief $\posterior$,
and in total the agent pays cost $c$ multiplied by the number of times that experiment $\experiment'$ is conducted \citep[e.g, ][]{wald1945sequential}.
\end{itemize}

Note that although the general results in our paper do not require any additional assumption on the valuation function, 
we will consider the following class of valuation functions in \Cref{sec:opt} to obtain more structural results on the optimal mechanism.
Essentially, we will focus on the setting where the private type of the agent represents her value from acquiring additional information. 
In particular, the valuation function of the agent is linear. 
\begin{definition}\label{def:linear}
The valuation $\V(\posterior,\type)$ is \emph{linear} 
if there exists a function $\val(\posterior)$ 
such that
$\V(\posterior,\type) = \val(\posterior) \cdot \type$ 
for any posterior $\posterior$ and any type $\type$.
\end{definition}

Note that this linear valuation assumption does not trivialize the
problem since the induced valuation function
$\V(\posterior, \experiments, \type)$ is still not linear in type
$\type$.  
Next, we introduce two canonical settings that satisfy the linear valuation assumption. 

\begin{example}\label{exp:match action}
Consider the model of a decision maker who makes a prediction over states~$\states$.
In our paper, the agent is the decision maker who chooses an action from action space $\actions$
to maximize her payoff. 
There are several payoff functions of the decision maker that are commonly considered in the literature. 
\begin{itemize}
\item Matching utilities: in this case, the state space and action space are finite, and $\states = \actions = \{1,\dots,n\}$.
The agent gains positive utility if the chosen action matches the state, 
i.e., the utility of the agent 
is $u(\action,\state;\type) = \type \cdot \indicate{\action = \state}$, 
where $\indicate{\cdot}$ is the indicator function
and~$\type$ is the private type of the agent.\footnote{In the special case where $\type = 1$ with probability 1,
this utility function is the matching utility considered in \citet*{bergemann2018design}. 
Note that in \citet*{bergemann2018design}, the agent has an exogenous private signal that is informative about the state, 
while in our model, that private signal is assumed to be endogenous.}
Given belief $\posterior$,
when the agent chooses her action optimally, 
her expected utility is 
$\V(\posterior, \type) = \type \cdot \max_{\state\in\states} \posterior(\state)$. 
Thus, by letting $\val(\posterior) = \max_{\state\in\states} \posterior(\state)$, 
the valuation of the agent is linear 
and $\V(\posterior,\type) = \val(\posterior) \cdot \type$
for any posterior $\posterior$
and type $\type$.

\item Error minimization: in this case, $\states = \actions \subseteq \reals$, 
and the agent minimizes the square error between the chosen action and the true state, 
i.e., the utility of the agent 
is $u(\action,\state;\type) = - \type \cdot (\action - \state)^2$.
Given belief $\posterior$, the optimal choice of the agent is $\expect[]{\state}$,
and her expected utility 
$\V(\posterior, \type) = -\type \cdot \var(\posterior)$, 
where $\var(\posterior)$ is the variance of distribution $\posterior$. 
Thus, by letting $\val(\posterior) = -\var(\posterior)$, 
the valuation of the agent is linear 
and $\V(\posterior,\type) = \val(\posterior) \cdot \type$
for any posterior $\posterior$
and type $\type$.
\end{itemize}
\end{example}

\begin{example}\label{exp:firm auction}
Consider the model of monopoly auction in \citet*{mussa1978monopoly}.
In this example, the agent is a firm selling a product to a consumer with private values for different quality levels of the product.
The state space $\states = \reals_+$ represents the space of consumer valuations. 
The firm has private cost parameter~$c$,
and the cost of producing a product with quality $q$ is $c\cdot q^2$.\footnote{\citet*{yang2020selling} considers a similar setting with linear cost function $c\cdot q$.}
We first provide an informal argument by assuming that the distribution $\posterior$ is regular, 
i.e., the virtual value function $\virtual_{\posterior}(z) = z - \frac{1-F_{\posterior}(z)}{f_{\posterior}(z)}$ 
is nondecreasing in~$z$,
where $F_{\posterior}$ and $f_{\posterior}$ are the cumulative function 
and density function given posterior belief~$\posterior$. 
The optimal mechanism of the firm with cost $c$ is to provide the product with quality
$q(z) = \frac{\max\{0,\virtual_{\posterior}(z)\}}{2c}$
to the agent with value $z$,
and the expected profit of the firm is
\begin{align*}
\int_{\reals_+} \frac{\max\{0,\virtual_{\posterior}(z)\}^2}{4c} \dd \posterior(z).
\end{align*}
Let $\type = \frac{1}{c}$ be the private type of the firm, 
and let $\val(\posterior) = \frac{1}{4}\int_{\reals_+} \max\{0,\virtual_{\posterior}(z)\}^2 \dd \posterior(z)$.
The valuation function 
is $\V(\posterior,\type) = \val(\posterior)\cdot \type$
given any type $\type$ and any belief $\posterior$, 
which satisfies the linearity condition. 
Note that this argument is not formal since the posterior $\posterior$ induced by the experiment is endogenously determined by the principal, and $\posterior$ might be irregular. 
We provide a formal argument in \cref{apx:monopoly_exp}.
\end{example}

\section{Optimal Mechanisms}
\label{sec:opt}
In this section, we characterize the optimal mechanisms for the setting with linear valuation functions 
under a mild regularity assumption on the agent's type distribution. 
This regularity assumption has been widely adopted in the auction design literature since \citet*{myerson1981optimal}.
Let $\phi(\type) = \type - \frac{1-\dist(\type)}{\density(\type)}$ 
be the virtual value function of the agent. 
\begin{assumption}\label{asp:regular}
The distribution $\dist$ is \emph{regular},
i.e., the corresponding virtual value function $\phi(\type)$ is monotone nondecreasing in $\type$.
\end{assumption}

Let monopoly type $\monopt = \inf_{\type} \{\virtual(\type) \geq 0\}$
be the lowest type with virtual value $0$.

In this section, we characterize the optimal revenue of the data broker using the envelope theorem \citep*{milgrom2002envelope}. 
For an agent with private type $\type$, the interim utility given revelation mechanism $\mech=(\alloc,\price)$ is 
\begin{align*}
\Util(\type) = \expect[\posterior\sim\alloc(\type)\given\prior]{
\V(\posterior, \experiments, \type)} - \price(\type).
\end{align*}
The derivative of the utility is 
\begin{align*}
\Util'(\type) = \expect[\posterior\sim\alloc(\type)\given\prior]{
\V_3(\posterior, \experiments, \type)}
= \expect[\posterior\sim\alloc(\type)\given\prior]{
\expect[\hat{\posterior}\sim \experiment_{\type,\posterior}\given\posterior]{\val(\hat{\posterior})}}
\end{align*}
where $\experiment_{\type,\posterior} \in \experiments$ is the optimal experiment for the agent given private type $\type$ and belief $\posterior$
and $\V_3(\posterior, \experiments, \type)$ is the partial derivative on the third coordinate. 
Thus, we have 
\begin{align*}
\Util(\type) = \int_{\underline{\type}}^{\type} 
\expect[\posterior\sim\alloc(z)\given\prior]{
\expect[\hat{\posterior}\sim \experiment_{z,\posterior}\given \posterior]{\val(\hat{\posterior})}} \dd z
+ \Util(\underline{\type}).
\end{align*}
Then, the revenue of the data broker is 
\begin{align}
\Rev(\mech) &= \expect[\type\sim\dist]{
\expect[\posterior\sim\alloc(\type)\given\prior]{
\V(\posterior, \experiments, \type)} 
- \int_{\underline{\type}}^{\type} 
\expect[\posterior\sim\alloc(z)\given\prior]{
\expect[\hat{\posterior}\sim \experiment_{z,\posterior}\given \posterior]{\val(\hat{\posterior})}} \dd z
- \Util(\underline{\type})}\nonumber\\
&= \expect[\type\sim\dist]{\expect[\posterior\sim\alloc(\type)\given\prior]{
\V(\posterior, \experiments, \type) 
- \frac{1-\dist(\type)}{\density(\type)}\cdot
\expect[\hat{\posterior}\sim \experiment_{\type,\posterior}\given\posterior]{\val(\hat{\posterior})} }}
- \Util(\underline{\type})\nonumber\\
&= \expect[\type\sim\dist]{
\expect[\posterior\sim\alloc(\type)\given\prior]{
\virtual(\type) \cdot \expect[\hat{\posterior}\sim \experiment_{\type,\posterior}\given\posterior]{
\val(\hat{\posterior})}
- \costagent(\experiment_{\type,\posterior},\posterior)}}
- \Util(\underline{\type}),
\label{eq:rev}
\end{align}
where the second equality holds by integration by parts.
The next lemma provides sufficient and necessary conditions on the allocations such that the resulting mechanism is incentive compatible and individually rational.

\begin{restatable}{lemma}{lemincentive}\label{lem:incentive}
An allocation rule $\alloc$ can be implemented by an incentive compatible and individually rational mechanism 
that ensures the lowest type with utility $\Util(\underline{\type})$
if and only if for any $\type, \type'\in \types$,\footnote{If $\type < \type'$, 
we use $\int_{\type'}^{\type}$ to represent $-\int_{\type}^{\type'}$.}
\begin{align}
\int_{\type'}^{\type} \expect[\posterior\sim\alloc(z)\given\prior]{\expect[\hat{\posterior}\sim \experiment_{z,\posterior}\given \posterior]{\val(\hat{\posterior})}} 
- \expect[\posterior\sim\alloc(\type')\given\prior]{\expect[\hat{\posterior}\sim \experiment_{z,\posterior}\given \posterior]{\val(\hat{\posterior})}} \dd z 
\geq 0,& \tag{IC}\\
\int_{\underline{\type}}^{\type} 
\expect[\posterior\sim\alloc(z)\given\prior]{
\expect[\hat{\posterior}\sim \experiment_{z,\posterior}\given \posterior]{\val(\hat{\posterior})}} \dd z
+ \Util(\underline{\type})
\geq \V(\prior, \experiments, \type).&\tag{IR}
\end{align}
\end{restatable}
The proof of \cref{lem:incentive} is deferred to \cref{apx:opt}.
The incentive constraint on allocation is similar to the integral monotonicity provided in \citet*{yang2020selling},
where the author considers selling data to an agent without any ability to further acquire information.
The IC condition in \cref{lem:incentive} is stronger than the monotonicity condition in allocation where high types are required to receive experiments that are Blackwell more informative than low types. 
In our model, monotonicity in allocation is not sufficient for designing incentive compatible mechanisms.\footnote{\citet*{sinander2019converse} shows that under some regularity conditions, 
monotone non-decreasing (in Blackwell order) allocations can always be implemented by
an incentive compatible mechanism. 
However, those conditions are violated in our paper, and there exists monotone allocations that cannot be implemented in any incentive compatible mechanism. 
One such counterexample is provided in \citet*{yang2020selling} under non-linear values.
} 
This illustrates a distinction between our model and the classical single-item auction, 
where in the latter case the monotonicity of the interim allocation is sufficient to ensure the incentive compatibility of the mechanisms.

\subsection{Without Endogenous Information}\label{sec:price}
As a preliminary, we first characterize the optimal mechanism in the simple case 
where the agent cannot acquire any additional information, i.e., $\experiments = \{\noexp\}$. 
In this case, we have $\costagent(\experiment,\prior) = 0$ for any $\experiment\in \experiments$ and any prior~$\prior$, 
and $\expect[\hat{\posterior}\sim \experiment_{\type,\posterior}\given\posterior]{
\val(\hat{\posterior})} 
= \val(\posterior)$. 
Thus, the revenue of any incentive compatible mechanism $\mech$ in \Cref{eq:rev} simplifies to 
\begin{align*}
\Rev(\mech) =
\expect[\type\sim\dist]{
\virtual(\type) \cdot 
\expect[\posterior\sim\alloc(\type)\given\prior]{
\val(\posterior)}}
- \Util(\underline{\type}).
\end{align*}
Intuitively, in this setting, 
by viewing $\expect[\posterior\sim\alloc(\type)\given\prior]{
\val(\posterior)}$
as a single-dimensional allocation,
the incentive constraint simplifies to the monotonicity constraint on quantity $\expect[\posterior\sim\alloc(\type)\given\prior]{
\val(\posterior)}$,
and hence the allocation in the optimal mechanism is a step function \citep*{myerson1981optimal}, 
which in our setting corresponds to revealing full information if $\phi(\type)\geq 0$
and revealing no information if $\phi(\type) < 0$. 
Note that this is equivalent to posting
a deterministic price for revealing full information. 
This intuition is formalized in \cref{prop:no acquisition}, with the
proof being deferred to \cref{apx:opt}.

\begin{restatable}{proposition}{propnoacquisition}\label{prop:no acquisition}
For linear valuations, under \cref{asp:concatenation} and \ref{asp:regular},
if the set of possible experiments is a singleton, 
i.e., the agent cannot acquire any additional information,
the optimal mechanism is to post a deterministic price for revealing full information. 
\end{restatable}
\cref{prop:no acquisition} shows that the optimal mechanism has a simple form when the agent cannot acquire endogenous information. 
This implies that, for example, 
when the data broker sells information to a firm that supplies
products to consumers with a quadratic cost for quality (c.f., \cref{exp:firm auction}), 
the optimal mechanism is posted pricing. 
This is in contrast to \citet*{yang2020selling}, 
where the firm has a linear cost function.
The optimal mechanism there is a  $\varphi$-quasi-perfect mechanism, 
which is considerably more complicated than posted pricing.

\subsection{Endogenous Information}

When agents can acquire endogenous information, the incentive constraints in \cref{lem:incentive} cannot be simplified to the monotonicity constraint, 
and the individual rationality constraint may bind for types higher than the lowest type $\underline{\type}$. 
Thus the pointwise optimization method for classical auction design cannot be applied when there is endogenous information. 
In the following theorem, we provide a full characterization of the optimal mechanism under \cref{asp:concatenation} and \ref{asp:regular}
by directly addressing the constraints on the integration of allocations. 
The detailed proof of \cref{thm:optimal linear} is provided in \cref{apx:opt}.

\begin{restatable}{theorem}{thmoptlinear}\label{thm:optimal linear}
For linear valuations, under \cref{asp:concatenation} and \ref{asp:regular},
there exists an optimal mechanism $\widehat{\mech}$ with allocation rule $\hat{\alloc}$ such that,\footnote{The characterization of allocation actually holds in any optimal mechanism except for a set of types with measure zero.} 
\begin{itemize}
    \item for any type $\type \geq \monopt$, the data broker reveals full information, i.e., $\hat{\alloc}(\type) = \fullexp$;
    
    \item for any type $\type < \monopt$, the data broker commits to experiment
        \begin{align*}
        \hat{\alloc}(\type) =
        \arg\max_{\experiment\in\experiments} \,
        \expect[\posterior \sim \experiment\given\prior]{\V(\posterior, \type)}
        - \costagent(\experiment,\prior)
        \end{align*}
        where ties are broken by maximizing the cost $\costagent(\experiment,\prior)$;
        
    \item $\Util(\underline{\type})
    = \V(\prior, \experiments, \underline{\type})$.
\end{itemize}

\end{restatable}

First, \cref{thm:optimal linear} implies that 
there is no distortion at the top
in the optimal mechanism. 
Intuitively, when the agent has sufficiently high type, 
i.e., $\virtual(\type) > 0$, 
by fully revealing the information to the agent, 
the expected virtual value is maximized 
since 
$\expect[\posterior\sim\signal\given\prior]{\expect[\hat{\posterior}\sim \experiment_{\type,\posterior}\given\posterior]{\val(\hat{\posterior})}}$
is maximized when the signal $\signal$ fully reveals the state. 
Moreover, with a fully revealed state, 
the posterior belief of the agent is a singleton, 
and hence the cost of the endogenous information is zero 
since there is no additional information available.
By \Cref{eq:rev}, this allocation maximizes the virtual surplus and hence the expected revenue of the data broker.

Moreover, according to the characterization in \cref{thm:optimal linear}, 
for any type $\type < \monopt$, the utility of the agent in the optimal mechanism $\widehat{\mech}$ is 
\begin{align*}
\Util(\type) &= \int_{\underline{\type}}^{\type} 
\expect[\posterior\sim\hat{\alloc}(z)\given\prior]{
\expect[\hat{\posterior}\sim \experiment_{z,\posterior}\given \posterior]{\val(\hat{\posterior})}} \dd z
+ \Util(\underline{\type})\\
&= \int_{\underline{\type}}^{\type} 
\expect[\posterior\sim\experiment_{\type,\prior}\given\prior]{
\val(\posterior)} \dd z
+ \V(\prior, \experiments, \underline{\type}) 
= \V(\prior, \experiments, \type).
\end{align*}
Thus, the individual rationality constraint binds not only for the lowest type 
but also for all types below the monopoly type $\monopt$.
Note that this is different from the classic auction design problem or selling information when the agent cannot acquire additional information. 
In those cases, the utilities of the low-type agents coincides with their outside options because the seller chooses not to sell to these agents. 
Here, the data broker provides valuable information to the agent for a positive payment such that low-type agents are exactly indifferent between participation
and opting out.

\section{Pricing for Full Information}
\label{sec:pricing}
The characterization in \cref{thm:optimal linear} demonstrates that
the optimal mechanism may contain a continuum of menu options,
which discriminate different types of the agent by offering
experiments with an increasing level of informativeness. 
Alternatively, a simpler mechanism that the principal can adopt in practice is pricing for full information, 
which does not require the principal to have the commitment power to randomize and conceal partial information. 
This simple mechanism is considerably easier and potentially less costly for the principal to implement in practice. 
In this section, we show that the additional benefit of price discrimination is limited, 
as posting a deterministic price for revealing full information is approximately optimal for revenue maximization
given the same set of assumptions as in \cref{thm:optimal linear}.
The proof of \cref{thm:approx} is provided in \cref{apx:pricing}.

\begin{restatable}{theorem}{thmapprox}\label{thm:approx}
For linear valuations, under \cref{asp:concatenation} and \ref{asp:regular},
for any prior $\prior$ and any cost function~$\costagent$,
posting a deterministic price for revealing full information
achieves at least half of the optimal revenue. 
\end{restatable}

Intuitively, the revenue in the optimal mechanism can be decomposed into two parts, 
the revenue contribution from high types (types above $\monopt$)
and the revenue contribution from low types (types below~$\monopt$).
The revenue contribution from high types can be directly covered by the optimal revenue from selling full information, 
since in the optimal mechanism, all those types receive full information and pay the threshold price, 
i.e., the price such that the cutoff type $\monopt$ is indifferent between acquiring full information from the principal and opting out of the auction.
Thus, it is sufficient to bound the revenue contribution from low types. 

Fixing any optimal price for selling full information
and the probability that the agent accepts the price, 
we can provide a lower bound on the cumulative distribution function (CDF) evaluated at any low type.
This is because if the CDF of any low type is lower than our bound,
the principal could adjust the price and correspondingly improve the revenue from selling full information,
a contradiction to the optimality of the given price. 
Since the optimal revenue of the seller is monotone in first-order stochastic dominance of the type distribution, 
given this lower bound on the CDF, 
we can derive an upper bound on the total revenue extracted from low types, 
which can further be bounded by the optimal revenue from pricing for full information. 
The main intuition behind this bound is that the worst case distribution over types cannot be too differentiated,
and hence the comparative advantage of screening compared to pricing is limited.\footnote{
In comparison, \citet{bergemann2021is} demonstrate that if the agent receives an exogenous signal before the auction, 
for any small constant $c > 1$ and any large number $m > 1$, 
there exists a valuation function of the agent and a signal structure such that 
any mechanism that is a $c$-approximation ratio to the optimal revenue
has menu size of at least $m$.}

Note that the loss of half of the optimal revenue is derived in the worst-case analysis. 
In many practical applications, the type distributions and cost functions may not coincide with the worst case, 
and the gap between the optimal revenue and pricing for full information can be significantly smaller than 2 (see \cref{sec:application}).
This provides further justification for adopting the simple pricing mechanisms instead of the optimal mechanism in practice.

In the remainder of this section,
we will also provide sufficient conditions such that
pricing for revealing full information is exactly optimal. 

\begin{restatable}{proposition}{propinformativeprior}\label{prop:informative prior}
\sloppy
For linear valuations, under \cref{asp:concatenation} and \ref{asp:regular},
for any cost $\costagent$ and any prior $\prior$,
the optimal mechanism is to post a price for revealing full information
if $\noexp \in \arg\max_{\experiment\in\experiments} \,
\expect[\hat{\posterior} \sim \experiment\given\prior]{\V(\hat{\posterior}, \monopt)}
- \costagent(\experiment,\prior)$. 
\end{restatable}
In \cref{apx:pricing}, we will show that when the monopoly type $\monopt$ finds it optimal to not acquire costly information, 
all types below $\monopt$ will have strict incentives to not acquire costly information. 
Combining this observation with the characterization in \cref{thm:optimal linear}, 
we directly obtain the result that the optimal mechanism is pricing for revealing full information.

Note that the condition in \cref{prop:informative prior} is also necessary for pricing for revealing full information to be revenue-optimal 
when $\monopt > \underline{\type}$ and the type distribution is continuous with positive density everywhere. 
This is because if the monopoly type $\monopt$ has a strict incentive to acquire costly information given the prior, 
there exists a positive measure of types below $\monopt$
that also have a strict incentive to acquire costly information given the prior. 
Hence, the data broker can obtain a strict revenue increase by price
discriminating those low types.

The interpretation of \cref{prop:informative prior}
is that discrimination and offering partial information are not beneficial when the prior is sufficiently informative. 
Fixing the cost function $\costagent$, we say the prior $\prior$ is \emph{sufficiently informative} if 
$\noexp \in \arg\max_{\experiment\in\experiments} \,
\expect[\hat{\posterior} \sim \experiment\given\prior]{\V(\hat{\posterior}, \monopt)}
- \costagent(\experiment,\prior)$.
This is intuitive since when the prior is sufficiently close to the degenerate pointmass distribution, 
the marginal cost of additional information is sufficiently high while the marginal benefit of additional information is bounded. 
Thus, it is not beneficial for the agent to not acquire any additional information.
We formalize the intuition in \cref{apx:pricing}.

\section{Applications}
\label{sec:application}
In this section, we apply the characterizations of the optimal mechanisms to several leading examples of selling information.

\vspace{10pt}\noindent
\textbf{Error Minimization.}
Here we consider a model where the agent is a decision maker attempting to minimize the square error of the chosen action. 
That is, let the state space and action space be $\states = \actions \subseteq \reals$, 
and the agent minimizes the square error between the chosen action and the true state, 
i.e., the utility of the agent 
is $u(\action,\state;\type) = - \type \cdot (\action - \state)^2$.
This is one of the models 
illustrated in \cref{exp:match action}. 

Recall that in \cref{exp:match action}, 
we show that the valuation function of the agent is linear 
with the form 
$\V(\posterior, \type) = \type \cdot \val(\posterior)$, 
where $\val(\posterior)=-\var(\posterior)$ is the variance of distribution $\posterior$. 
Let~$\dist$ be the distribution over the types,
and let $\monopt$ be the monopoly type in distribution~$\dist$.
We assume that the prior distribution $\prior$ over states is a Gaussian distribution $\g(0, \eta^2)$ with variance $\eta^2$. 
The agent can repeatedly pay a unit cost $c$
to observe a Gaussian signal 
$s = \state + \epsilon$
where $\epsilon\sim \g(0,1)$.

Next, we illustrate the optimal mechanism in this setting by applying \cref{thm:optimal linear}. 
\begin{itemize}
\item For any $\type\geq \monopt$, 
the data broker reveals the states to the firm 
at price $p = \monopt\cdot\var(\prior)$.

\item For any $\type < \monopt$, 
the optimal allocation solves the following information acquisition problem: 
\begin{align*}
\hat{\alloc}(\type) &=
\arg\max_{\experiment\in\experiments} \,
\expect[\posterior \sim \experiment\given\prior]{\V(\posterior, \type)}
- \costagent(\experiment,\prior).
\end{align*}
Note that in this example, the agent can only decide the number of Gaussian signals to observe, 
and with $k$ signals, 
the cost is $k$, 
and the variance of the posterior is $\frac{\eta^2}{1+k\eta^2}$ regardless of the realized sequence of the observed signals.
Thus, letting 
\begin{align*}
k_{\type} = \argmax_k - \type \cdot \frac{\eta^2}{1+k\eta^2} - kc,
\end{align*}
in the optimal mechanism, the data broker commits to an experiment that is Blackwell equivalent to revealing $k_{\type}$
Gaussian signals with unit variance and charges the agent price $k_{\type}\cdot c$.
\end{itemize}

Note that for any $\type < \monopt$, 
the optimal number of signals revealed to the agent $k_{\type}$ 
is weakly increasing in~$\type$ and $\eta$ 
and is weakly decreasing in $c$. 
Moreover, fixing distribution $\dist$ and correspondingly the monopoly type $\monopt$, 
when $\eta$ is sufficiently small or when $c$ is sufficiently large, 
$k_{\type} = 0$ for any $\type < \monopt$, 
and the optimal mechanism reduces to the posted pricing mechanism.

Finally, to compare the revenue between the optimal mechanisms and pricing for full information in this setting, 
assume that $\dist$ is the uniform distribution in $[0,1]$. 
Consider the choice of parameters $c\in\{0.05,0.1,0.2\}$
and $\eta^2\in\{0.2,0.5,0.8\}$,  and
the revenues for the two mechanisms are computed in \cref{tab:err min}. 
As shown in the table, the maximum ratio between the two mechanisms given the choice of the parameter sets is at most $1.13$.
\begin{table}[t]
    \centering
    \begin{tabular}{|c|c|c|c|}
       \hline
         & 0.2 & 0.5 & 0.8 \\
       \hline
       0.05  & 0.05 / 0.05 & 0.118 / 0.110 & 0.152 / 0.135  \\
       \hline
       0.1  & 0.05 / 0.05 & 0.125 / 0.125 & 0.183 / 0.167  \\
       \hline
       0.2  & 0.05 / 0.05 & 0.125 / 0.125 & 0.2 / 0.2  \\
       \hline
    \end{tabular}
    \caption{In this table, each column represents a choice of parameter $\eta^2$ and each row represents a choice of parameter $c$.
    The first number in each entry is the revenue of the optimal mechanism, and the second number is the revenue from pricing for full information.}
    \label{tab:err min}
\end{table}

\vspace{10pt}\noindent
\textbf{Monopoly auction.}
Here, we consider the monopoly auction model introduced in \citet*{mussa1978monopoly}. 
This is introduced in \cref{exp:firm auction}. 
We consider a simple case in which the state space $\states = \{\state_1,\state_2\}$ is binary, 
where $\state_i\in\reals$ represents the value of the consumer
and $0< \state_1 < \state_2$.
In this case, given posterior belief $\posterior$ of the firm, 
the virtual value of the consumer simplifies to 
\begin{align*}
\virtual_{\posterior}(\state_1) &= \state_1 - \frac{\posterior(\state_2)(\state_2-\state_1)}{\posterior(\state_1)};\\
\virtual_{\posterior}(\state_2) &= \state_2.
\end{align*}
According to \citet*{mussa1978monopoly},
the optimal mechanism of a firm with cost $c$ is to provide the product with quality
$q(\state_i) = \frac{\max\{0,\virtual_{\posterior}(\state_i)\}}{2c}$
to the agent with value $\state_i$,
and the expected profit of the firm is $\frac{1}{c}\cdot\val(\posterior)$, 
where 
\begin{align*}
\val(\posterior) \triangleq \posterior(\state_1) \cdot
\frac{\max\{0,\virtual_{\posterior}(\state_1)\}^2}{4}
+ \posterior(\state_2) \cdot \frac{\state^2_2}{4}.
\end{align*}
Suppose that cost $c$ is the private information of the firm,
and let $\type = 1/c$.
Recall that $\dist$ is the distribution over type $\type$
and $\prior$ is the prior over states.
Let $\monopt$ be the monopoly type in distribution $\dist$. 
We assume that the firm can flexibly design any experiment, i.e., $\experiments$ contains all possible experiments. 
In addition, for any $\experiment\in\experiments$,
the cost is the reduction in entropy, 
i.e., 
$\costagent(\experiment,\posterior) = \entropy(\posterior)
- \expect[\hat{\posterior}\sim\experiment\given\posterior]{\entropy(\hat{\posterior})}$ 
for any posterior $\posterior$,
where $\entropy(\posterior) = -\sum_{i} \posterior(\state_i) \log \posterior(\state_i)$
is the entropy function. 

Since the valuation function of the firm is linear, 
next we illustrate the optimal mechanism in this setting by applying \cref{thm:optimal linear}. 
\begin{itemize}
\item For any $\type\geq \monopt$, 
or equivalently for any $c \leq 1/\monopt$,
the data broker reveals full information to the firm. 
\begin{figure}[t]
\centering
\input{fig/auction}
\caption{\label{fig:auction} 
The figure is the value of $\val(\hat{\posterior}) \cdot \type + \entropy(\hat{\posterior})$ 
as a function of $\hat{\posterior}(\state_1)$.
}
\end{figure}
\item For any $\type < \monopt$, 
or equivalently for any $c > 1/\monopt$, 
the optimal allocation solves the following Bayesian persuasion problem: 
\begin{align*}
\hat{\alloc}(\type) &=
\arg\max_{\experiment\in\experiments} \,
\expect[\posterior \sim \experiment\given\prior]{\V(\posterior, \type)}
- \costagent(\experiment,\prior) \\
&= \arg\max_{\experiment\in\experiments} \,
\expect[\posterior \sim \experiment\given\prior]{\type \cdot\val(\posterior)
+ \entropy(\posterior)}
- \entropy(\prior).
\end{align*}
By the concavification approach in \citet*{kamenica2011bayesian}, the optimal experiment has signal space of size 2. 
As illustrated in \cref{fig:auction}, if the prior satisfies
$\posterior_{\type}(\state_1) < \prior(\state_1) < \posterior'_{\type}(\state_1)$, 
the data broker induces either posterior $\posterior_{\type}$ or $\posterior_{\type}'$ for type~$\type$. 
Otherwise, the data broker reveals no information to the firm. 
\end{itemize}

To compare the revenue between the optimal mechanism and pricing for full information in this setting, 
we assume that type $\type$ is drawn from a uniform distribution in $[0,10]$.
Moreover, we consider binary state space $\states = \{1,2\}$
and measure the gap based on the choice of prior $\prior$. 
The revenues of the two mechanisms are computed in \cref{tab:monopoly}. 
As shown in the table, the maximum ratio between the two mechanisms given the choice of parameter sets is at most $1.10$.

\begin{table}[t]
    \centering
    \begin{tabular}{|c|c|c|c|}
       \hline
        0.2 & 0.4 & 0.6 & 0.8 \\
       \hline
       0.125 / 0.125  & 0.244 / 0.235 & 0.294 / 0.269 & 0.211 / 0.200  \\
       \hline
    \end{tabular}
    \caption{In this table, the first row is the probability of state $1$ given prior $\prior$, 
    and the second row is the revenue from the optimal mechanisms and pricing for full information.}
    \label{tab:monopoly}
\end{table}

\bibliographystyle{plainnat} 
\bibliography{ref} 
\newpage

\begin{appendix}
\section{General Valuations}\label{apx:general}

In this section, we will provide characterizations of the revenue-optimal mechanism without the linearity assumption on the valuation function. 

\begin{restatable}{theorem}{thmoptimal}\label{thm:optimal}
Under \cref{asp:concatenation}, assuming that there exists $\beta > 0$
such that $|\V(\posterior, \type)| \leq \beta$ for any $\posterior$ and any~$\type$ in the support of $\dist$,
there exists a revenue-optimal mechanism such that the following two properties hold. 
\begin{enumerate}
\item The agent does not acquire costly information in equilibrium. 
That is, 
\begin{align*}
\expect[\type\sim\dist]{\expect[\posterior\sim\alloc(\type)\given\prior]{
\costagent(\experiment^*_{\posterior,\type}, \posterior)}} = 0
\end{align*}
where $\experiment^*_{\posterior,\type} \in \arg\max_{\experiment\in\experiments} \,
\expect[\hat{\posterior} \sim \experiment\given\posterior]{\V(\hat{\posterior}, \type)}
- \costagent(\experiment,\posterior)$.\footnote{Since $\costagent(\experiment,\posterior)\geq 0$ for any $\experiment$ and $\posterior$, 
the agent only acquires costly information for a set with measure zero.
}

\item Revealing full information is in the menu of the optimal mechanism, 
i.e., there exists a type $\type\in\types$ such that $\alloc(\type) = \fullexp$.
\end{enumerate}
\end{restatable}
Before we formally state the proofs, we first interpret our results. 
The second statement of \cref{thm:optimal} is analogous to the standard no distortion at the top result. 
As we made no assumption on the type space here, 
unlike a setting with linear valuations where there is a natural order on the type space, 
we cannot pin down the type that receives full information in the optimal mechanism. 
What is more interesting is the first statement, 
where the theorem states that in equilibrium, 
the agent never (except for a set with measure zero) has an incentive to acquire additional costly information after receiving the signal from the data broker. 
This holds because if the agent with type~$\type$ acquires additional
information by conducting experiment $\experiment_{\type}$ at a positive cost, 
the data broker can directly supply this experiment to the agent
and increases the payment of type $\type$ by the cost of experiment $\experiment_{\type}$.
The new mechanism increases the expected revenue of the data broker 
and eliminates the incentives for an agent with type $\type$ to further acquire any additional information.
In the following proof, we will formally show that this new mechanism is also incentive compatible and individually rational.\footnote{Note that this result relies crucially on \cref{asp:concatenation}. 
In \cref{apx:necessity}, we will show that if \cref{asp:concatenation} is violated, 
it is possible that the agent has a strict incentive to acquire additional costly information in equilibrium.} 

\begin{proof}[Proof of Statement 1 of \cref{thm:optimal}]
Let $\mech = (\alloc,\price)$ be the optimal mechanism.
Let $\kappa_{\type,\signal}$ be the optimal choice of experiments for an agent with type $\type$
when she receives the realized experiment $\signal$
from mechanism~$\mech$.
By contradiction, 
let $\hat{\types}$ be the set of types with positive measure such that 
for any $\hat{\type} \in \hat{\types}$,
the cost of additional experiments given the optimal best response strategy $\kappa_{\hat{\type},\signal}$
for an agent with type $\hat{\type}$
is positive, 
i.e., 
\begin{align*}
\int_{\states} \int_{\signalspacea} 
\costagent(\kappa_{\hat{\type},\alloc(\hat{\type})}(\sa), \hat{\posterior}_{\alloc(\hat{\type}),\sa,\prior})
\dd \alloc(\hat{\type})(\sa\given \state) \dd \prior(\state) 
> 0,
\end{align*}
where $\hat{\posterior}_{\signal,\sa,\prior}$ is the posterior given experiment $\signal$ and signal $\sa$,
assuming that the prior is $\prior$.
Let~$\hat{\alloc}$ and $\hat{\price}$ be the allocation and payment rule such that 
\begin{itemize}
\item for any $\type\not\in \hat{\types}$, 
$\hat{\alloc}(\type) = \alloc(\type)$ and 
$\hat{\price}(\type) = \price(\type)$;

\item for any $\hat{\type}\in \hat{\types}$,
for any $\signal\in\signals$,
$\hat{\alloc}(\hat{\type}) = \kappa_{\hat{\type},\signal} \circ \alloc(\hat{\type})$,\footnote{Note that under this new sequential experiment $\hat{\alloc}(\hat{\type})$,
the signals generated by experiments in all stages will be revealed to the agent.}
and 
\begin{align*}
\hat{\price}(\hat{\type}) = 
\price(\hat{\type}) 
+ \int_{\states} \int_{\signalspacea} 
\costagent(\kappa_{\hat{\type},\alloc(\hat{\type})}(\sa), \hat{\posterior}_{\alloc(\hat{\type}),\sa,\prior})
\dd \alloc(\hat{\type})(\sa\given \state) \dd \prior(\state). 
\end{align*}
\end{itemize}
Let $\hat{\mech} = (\hat{\alloc}, \hat{\price})$. 
It is easy to verify that 
\begin{align*}
\Rev(\hat{\mech}) 
&= \int_{\types} \hat{\price}(\type) \dd\dist(\type)
= \int_{\types \setminus \hat{\types}} \hat{\price}(\type) \dd\dist(\type) 
+\int_{\hat{\types}} \hat{\price}(\type) \dd\dist(\type) \\
&> \int_{\types \setminus \hat{\types}} 
\price(\type) \dd\dist(\type)
+ \int_{\hat{\types}} 
\price(\type) \dd\dist(\type)
= \Rev(\mech).
\end{align*}
The inequality holds because the types in set $\hat{\types}$ occur with positive measure
and for any type $\type\in \hat{\types}$, 
$\hat{\price}(\hat{\type}) >
\price(\hat{\type})$. 
Thus, the revenue of mechanism $\hat{\mech}$ is strictly higher.
Moreover, in mechanism $\hat{\mech}$, 
the utility of the agent
has at least the same expected utility as in mechanism $\mech$ 
by not acquiring any additional information upon receiving the signal. 
Therefore, mechanism $\hat{\mech}$ is individually rational.
Next, it is sufficient to show that mechanism $\hat{\mech}$ is incentive compatible. 
It is easy to verify that an agent with any type $\type$ has no incentive to deviate to type $\hat{\type}\not\in\hat{\types}$ 
since her utility from reporting truthfully weakly increases, 
while her utility from misreporting $\hat{\type}$ remains unchanged. 
Finally, for any type $\type$, 
under mechanism $\hat{\mech}$, 
the utility from deviating the report from type $\type$ to type $\hat{\type}\in \hat{\types}$ is
\begin{align*}
&\Util(\type; \hat{\type},\hat{\mech})
= \expect[\posterior\sim\hat{\alloc}(\hat{\type})\given\prior]{
\V(\posterior, \experiments, \type)}
- \hat{\price}(\hat{\type})\\
&= \expect[\posterior\sim\hat{\alloc}(\hat{\type})\given\prior]{
\V(\posterior, \experiments, \type)}
- \int_{\states} \int_{\signalspacea} 
\costagent(\kappa_{\hat{\type},\alloc(\hat{\type})}(\sa), \hat{\posterior}_{\alloc(\hat{\type}),\sa,\prior})
\dd \alloc(\hat{\type})(\sa\given \state) \dd \prior(\state) 
- \price(\hat{\type})\\
&\leq \int_{\states} \int_{\signalspacea} 
\V(\hat{\posterior}_{\alloc(\hat{\type}),\sa,\prior}, \experiments, \type)
\dd \alloc(\hat{\type})(\sa\given \state) \dd \prior(\state)
- \price(\hat{\type})\\
&= \Util(\type; \hat{\type},\mech).
\end{align*}
The inequality holds because by \cref{asp:concatenation},
given any posterior $\hat{\posterior}_{\alloc(\hat{\type}),\sa,\prior}$, 
a feasible choice for the agent is to first choose experiment $\kappa_{\hat{\type},\alloc(\hat{\type})}(\sa)\in \experiments$ 
with cost $\costagent(\kappa_{\hat{\type},\alloc(\hat{\type})}(\sa), \hat{\posterior}_{\alloc(\hat{\type}),\sa,\prior})$ 
and then choose additional experiments optimally given the realized signal. 
The utility of this choice is bounded from above by directly choosing the optimal experiment from $\experiments$, 
which induces value $\V(\hat{\posterior}_{\alloc(\hat{\type}),\sa,\prior}, \experiments, \type)$
for the agent.
Thus, we have 
\begin{align*}
\Util(\type; \hat{\mech}) - \Util(\type; \hat{\type},\hat{\mech})
\geq \Util(\type; \mech) - \Util(\type; \hat{\type},\mech) \geq 0,
\end{align*}
and mechanism $\hat{\mech}$ is incentive compatible.
\end{proof}

\begin{proof}[Proof of Statement 2 of \cref{thm:optimal}]
For any mechanism $\mech = (\alloc,\price)$,
let $\bar{\price} = \sup_{\type}\price(\type)$,
which is well defined since the values of the agent are bounded. 
By adding the choice $(\fullexp, \bar{\price})$ into the menu of mechanism $\mech$, 
the revenue of the data broker only increases.
\end{proof}

Another natural question for settings with a general valuation function is whether the approximate optimality of pricing for full information extends. 
Unfortunately, without a simple closed-form characterization of the optimal mechanism, it is challenging to provide a conclusive answer to this question.
Therefore, we leave it as an interesting open question for future work.

\subsection{Necessity of Assumption \ref{asp:concatenation}}
\label{apx:necessity}
In this section, we show that \cref{asp:concatenation} is necessary for \cref{thm:optimal}.
We will prove this by constructing an instance where the set of possible experiments for the agent violates \cref{asp:concatenation}. 
Note that
the goal of the following contrived construction is not to provide an instance that fits the economic applications. 
It only serves the purpose of showing that \cref{asp:concatenation}
is necessary for our observation in \cref{thm:optimal}.

Suppose that there is a binary state $\states = \{\state_1,\state_2\}$.
Thus, the posterior is uniquely determined by the probability of $\state_1$, 
and thus we will also use $\posterior$ to represent $\posterior(\state_1)$. 
We assume that the prior $\prior = \frac{1}{2}$.
The data broker can choose any experiment, 
and the set of possible experiments for the agent is 
$\{\noexp,\fullexp,\experiment_1\}$.
Recall that $\noexp$ reveals no information and $\fullexp$ reveals full information.
The signal space $\signalspace$ for $\experiment_1$ is $\{s_1,\s_2\}$,
and $\experiment_1(s_1\given \state_1) =\experiment_1(s_2\given \state_2) = \frac{2}{3}$.
The cost function satisfies $\cost(\noexp,\posterior) = 0$
and $\cost(\experiment_1,\posterior) = 1$ for all $\posterior$,
and 
\begin{align*}
\cost(\fullexp,\posterior) &= \begin{cases}
10 & \posterior \in [\frac{1}{3},\frac{2}{3}];\\
1 & \text{otherwise.}
\end{cases}
\end{align*}
We assume that the type space is also binary, i.e., $\types = \{\type_1,\type_2\}$.
There is common prior over the types $\dist$ where $\dist(\type_1) = \dist(\type_2) = \frac{1}{2}$.
We assume that 
\begin{align*}
\V(\posterior,\type_1) &= \begin{cases}
3 & \posterior = \frac{4}{5} \text{ or } \frac{1}{5};\\
0 & \text{otherwise;}
\end{cases}\\
\V(\posterior,\type_2) &= \begin{cases}
10 & \posterior = 0 \text{ or } 1;\\
0 & \text{otherwise.}
\end{cases}
\end{align*}
It is not difficult to verify that in the optimal mechanism $\mech = (\alloc,\price)$, and
we have that $\alloc(\type_1) = \experiment_1, \price(\type_1) = \frac{2}{3}$
and $\alloc(\type_2) = \fullexp, \price(\type_2) = 10$. 
In equilibrium, for type $\type_1$, the agent has an incentive to acquire additional information by conducting experiment $\experiment_1$ regardless of the signal realization. 
In the proof of \cref{thm:optimal}, 
if \cref{asp:concatenation} is satisfied, 
the data broker can simulate the behavior of $\type_1$ by offering
experiment $\experiment_1\circ\experiment_1$ to an agent with type $\type_1$, 
raise the payment for type $\type_1$,
and accordingly increase the expected revenue.
However, this option is not profitable in this constructed example. 
The main reason is that if the data broker offers type $\type_1$ experiment $\experiment_1\circ\experiment_1$, 
an agent with type $\type_2$ will have a strong incentive to deviate her report to $\type_1$, 
since the cost of conducting experiment $\fullexp$ is significantly reduced given the posterior distribution induced by experiment $\experiment_1\circ\experiment_1$. 
Thus, the payment extracted from type~$\type_2$ will be significantly smaller, and the expected revenue of the data broker is reduced in this case.

\section{Optimal Mechanisms}
\label{apx:opt}
Before the proof of the theorems in \cref{sec:opt}, 
we first present the following lemma showing that experiment~$\fullexp$ that reveals full information is the most valuable for the agent.
Recall that $\experiment_{\type,\posterior} \in \experiments$ is the optimal experiment 
the agent chooses when her type is $\type$
and her posterior belief after receiving the signal from the data broker is $\posterior$.

\begin{lemma}\label{lem:bound alloc}
Let $\fullsignal$ be the experiment that reveals full information. 
For any experiment $\signal\in\signals$, any prior $\prior$,
and any type $\type$, 
we have 
\begin{align*}
\expect[\posterior\sim\fullsignal\given\prior]{
\expect[\hat{\posterior}\sim \experiment_{\type,\posterior}\given\posterior]{
\val(\hat{\posterior})}} 
\geq \expect[\posterior\sim\signal\given\prior]{
\expect[\hat{\posterior}\sim \experiment_{\type,\posterior}\given\posterior]{
\val(\hat{\posterior})}}.
\end{align*}
\end{lemma}
\begin{proof}
For the fully informative experiment $\fullsignal$, 
for any experiment $\signal$, any prior $\prior$,
and any type $\type$, 
we have 
\begin{align*}
\expect[\posterior\sim\fullsignal\given\prior]{
\expect[\hat{\posterior}\sim \experiment_{\type,\posterior}\given\posterior]{
\val(\hat{\posterior})}} 
= \expect[\state\sim\prior]{
\val(\state)} 
\geq \expect[\posterior\sim\signal\given\prior]{
\expect[\hat{\posterior}\sim \experiment_{\type,\posterior}\given\posterior]{
\val(\hat{\posterior})}}. 
\end{align*}
The inequality holds since $\val$ is convex in $\posterior$ and $\fullsignal$ fully reveals the states.
\end{proof}

\lemincentive*
\begin{proof}[Proof of \cref{lem:incentive}]
Given allocation rule $\alloc$, 
by the envelope theorem, for any incentive compatible mechanism $\mech$,
the interim utility $\Util(\type)$ is convex in $\type$
and 
\begin{align*}
\Util(\type) = \int_{\underline{\type}}^{\type} 
\expect[\posterior\sim\alloc(z)\given\prior]{
\expect[\hat{\posterior}\sim \experiment_{z,\posterior}\given \posterior]{\val(\hat{\posterior})}} \dd z
+ \Util(\underline{\type}).
\end{align*}
Note that the mechanism $\mech = (\alloc,\price)$ is individually rational
if and only if $\Util(\type) \geq \V(\prior, \experiments, \type)$ for any type $\type$,
i.e., 
\begin{align*}
\int_{\underline{\type}}^{\type} 
\expect[\posterior\sim\alloc(z)\given\prior]{
\expect[\hat{\posterior}\sim \experiment_{z,\posterior}\given \posterior]{\val(\hat{\posterior})}} \dd z
+ \Util(\underline{\type})
\geq \V(\prior, \experiments, \type).
\end{align*}
Moreover, the corresponding payment rule for mechanism $\mech$ is 
\begin{align*}
\price(\type) = \expect[\posterior\sim\alloc(\type)\given\prior]{
\V(\posterior, \experiments, \type) 
- \int_{\underline{\type}}^{\type} 
\expect[\posterior\sim\alloc(z)\given\prior]{
\expect[\hat{\posterior}\sim \experiment_{z,\posterior}\given \posterior]{
\val(\hat{\posterior})}} \dd z} - 
\Util(\underline{\type}).
\end{align*}
Next we verify the incentive constraints of the given mechanism.
Note that for any $\type, \type'\in \types$,
letting $\Util(\type; \type')$ be the utility of the agent with type $\type$ when she reports $\type'$
in mechanism~$\mech$,
we have 
\begin{align*}
&\Util(\type) - \Util(\type; \type') \\
&= \Util(\type) - \Util(\type') 
- \expect[\posterior\sim\alloc(\type')\given\prior]{
\V(\posterior, \experiments, \type)}
+ \expect[\posterior\sim\alloc(\type')\given\prior]{
\V(\posterior, \experiments, \type')}\\
&= \int_{\type'}^{\type} 
\expect[\posterior\sim\alloc(z)\given\prior]{
\expect[\hat{\posterior}\sim \experiment_{z,\posterior}\given \posterior]{\val(\hat{\posterior})}} \dd z 
- \expect[\posterior\sim\alloc(\type')\given\prior]{
\int_{\type'}^{\type} \V_3(\posterior, \experiments, z) \dd z} \\
&= \int_{\type'}^{\type} \left(\expect[\posterior\sim\alloc(z)\given\prior]{\expect[\hat{\posterior}\sim \experiment_{z,\posterior}\given \posterior]{\val(\hat{\posterior})}} 
- \expect[\posterior\sim\alloc(\type')\given\prior]{\expect[\hat{\posterior}\sim \experiment_{z,\posterior}\given \posterior]{\val(\hat{\posterior})}}\right) \dd z.
\end{align*}
Thus, $\Util(\type) - \Util(\type; \type')\geq 0$
if and only if the integral constraint in the statement of \cref{lem:incentive} is satisfied.
\end{proof}

\propnoacquisition*
\begin{proof}
Let $\mech$ be the mechanism that reveals full information if $\phi(\type)\geq 0$
and reveals no information if $\phi(\type) < 0$.
When $|\experiments| = 1$, 
for any type $\type$ and any posterior $\posterior$,
we have that $\expect[\hat{\posterior}\sim \experiment_{\type,\posterior}\given\posterior]{
\val(\hat{\posterior})} 
= \val(\posterior)$,
and hence the integral constraint simplifies to 
\begin{align*}
&\int_{\type'}^{\type} \expect[\posterior\sim\alloc(z)\given\prior]{\expect[\hat{\posterior}\sim \experiment_{z,\posterior}\given \posterior]{\val(\hat{\posterior})}} 
- \expect[\posterior\sim\alloc(\type')\given\prior]{\expect[\hat{\posterior}\sim \experiment_{z,\posterior}\given \posterior]{\val(\hat{\posterior})}} \dd z \\
=& \int_{\type'}^{\type} \expect[\posterior\sim\alloc(z)\given\prior]{\val(\posterior)} 
- \expect[\posterior\sim\alloc(\type')\given\prior]{\val(\posterior)} \dd z 
\geq 0.
\end{align*}
Note that this is equivalent to the condition that 
$\expect[\posterior\sim\alloc(\type)\given\prior]{\val(\posterior)}$
is nondecreasing in $\type$. 
Moreover, similar to \citet*{myerson1981optimal},
the individual rationality constraint simplifies to the case in which the utility of the lowest type $\underline{\type}$ is at least her outside option. 
Thus, the problem of revenue maximization simplifies to 
\begin{align*}
\max & \quad
\expect[\type\sim\dist]{
\virtual(\type) \cdot 
\expect[\posterior\sim\alloc(\type)\given\prior]{
\val(\posterior)}}
- \Util(\underline{\type}) \\
\text{s.t.} & \quad
\expect[\posterior\sim\alloc(\type)\given\prior]{\val(\posterior)}
\text{ is non-decreasing in } \type,\\
&\quad \Util(\underline{\type}) \geq \V(\prior, \experiments, \underline{\type}).
\end{align*}
Note that $\expect[\posterior\sim\alloc(\type)\given\prior]{\val(\posterior)}$
is maximized by revealing full information 
and minimized by revealing no information. 
By \citet*{myerson1981optimal}, 
it is easy to verify that the allocation rule of mechanism $\mech$ maximizes the expected virtual surplus, 
and hence mechanism $\mech$ is the revenue-optimal mechanism.
\end{proof}

\thmoptlinear*
\begin{proof}[Proof of \cref{thm:optimal linear}]
We first show that allocation rule $\hat{\alloc}$ combined with $\Util(\underline{\type}) = \V(\prior, \experiments, \underline{\type})$ 
can be implemented as an incentive compatible and individually rational mechanism. 
One way to prove this is to verify the constraints specified in \cref{lem:incentive}. 
However, directly verifying the incentive constraints in \cref{lem:incentive} for allocation $\hat{\alloc}$
might be challenging because we impose little structure on the information costs.\footnote{Without additional structures on the costs, 
it is difficult to characterize the optimal strategy $\experiment_{\type,\posterior}$ given any type $\type$ and posterior $\posterior$.
} 
Thus, we adopt an alternative approach by explicitly constructing an incentive compatible and individually rational mechanism $\widehat{\mech}$. 
Then, we show that the constructed mechanism has allocation $\hat{\alloc}$ and utility for the lowest type 
$\Util(\underline{\type}) = \V(\prior, \experiments, \underline{\type})$.

First, consider a mechanism $\mech'$ that posts a deterministic price $p$ for revealing full information. 
Price $p$ is chosen such that the agent purchases information from the seller if and only if $\type \geq \monopt$. 
Note that given mechanism $\mech'$, for an agent with type $\type < \monopt$, 
she will choose not to participate the auction 
and subsequently conduct experiment
\begin{align*}
\experiment_{\type} =
\arg\max_{\experiment\in\experiments} \,
\expect[\posterior \sim \experiment\given\prior]{\V(\posterior, \type)}
- \costagent(\experiment,\prior).
\end{align*}
We assume that the agent breaks ties by maximizing the cost $\costagent(\experiment,\prior)$. 
Now, let $\widehat{\mech}$ be the mechanism that reveals full information for types $\type \geq \monopt$
with price $p$ 
and commits to experiment $\experiment_{\type}$
for types $\type < \monopt$
with price $\costagent(\experiment_{\type},\prior)$. 
It is easy to verify that $\widehat{\mech}$ has allocation rule $\hat{\alloc}$
and the utility of the lowest type $\underline{\type}$ in $\widehat{\mech}$ 
is $\V(\prior, \experiments, \underline{\type})$.
Moreover, by the proof of \cref{thm:optimal},
mechanism $\widehat{\mech}$ is incentive compatible and individually rational. 

Note that when the posterior $\posterior$
is in the support of $\fullsignal\given\prior$, 
the agent will not acquire additional costly information 
since $\fullsignal$ fully reveals the state. 
Moreover, when the posterior $\posterior$
is in the support of $\experiment_{\type}\given\prior$,
by \cref{asp:concatenation}, 
the agent will not acquire additional costly information
because otherwise $\experiment_{\type}$ is not the utility-maximizing experiment given prior $\prior$.
Combining the observations, we have that 
$\costagent(\experiment_{\type,\posterior},\posterior) = 0$ for $\posterior$
in the support of $\hat{\alloc}(\type)\given\prior$,
and hence by \Cref{eq:rev}, the revenue of mechanism $\widehat{\mech}$ is 
\begin{align}
&\Rev(\widehat{\mech}) = \expect[\type\sim\dist]{
\expect[\posterior\sim\hat{\alloc}(\type)\given\prior]{
\virtual(\type) \cdot \expect[\hat{\posterior}\sim \experiment_{\type,\posterior}\given\posterior]{
\val(\hat{\posterior})}
- \costagent(\experiment_{\type,\posterior},\posterior)}}
- \Util(\underline{\type}) \nonumber\\
&= \expect[\type\sim\dist]{
\virtual(\type) \cdot 
\expect[\posterior\sim\hat{\alloc}(\type)\given\prior]{
\expect[\hat{\posterior}\sim \experiment_{\type,\posterior}\given\posterior]{
\val(\hat{\posterior})}}}
- \V(\prior, \experiments, \underline{\type}) \nonumber\\
&= \int_{\monopt}^{\bar{\type}}
\density(\type)\cdot\virtual(\type) \cdot 
\expect[\posterior\sim \fullsignal\given\prior]{
\val(\posterior)}
\dd \type
+\int_{\underline{\type}}^{\monopt}
\density(\type)\cdot\virtual(\type) \cdot 
\expect[\posterior\sim \experiment_{\type,\prior}\given\prior]{
\val(\posterior)}
\dd \type
- \V(\prior, \experiments, \underline{\type}).\label{eq:opt rev}
\end{align}
Now, consider any incentive compatible and individually rational mechanism $\mech$ with allocation $\alloc$; 
again, by \Cref{eq:rev}, the revenue of mechanism $\mech$ is 
\begin{align*}
\Rev(\mech) &= \expect[\type\sim\dist]{
\expect[\posterior\sim\alloc(\type)\given\prior]{
\virtual(\type) \cdot \expect[\hat{\posterior}\sim \experiment_{\type,\posterior}\given\posterior]{
\val(\hat{\posterior})}
- \costagent(\experiment_{\type,\posterior},\posterior)}}
- \Util(\underline{\type})\\
&\leq \expect[\type\sim\dist]{
\virtual(\type) \cdot 
\expect[\posterior\sim\alloc(\type)\given\prior]{
\expect[\hat{\posterior}\sim \experiment_{\type,\posterior}\given\posterior]{
\val(\hat{\posterior})}}}
- \Util(\underline{\type}),
\end{align*}
where the inequality holds since $\costagent(\experiment_{\type,\posterior},\posterior) \geq 0$ for any posterior $\posterior$.
For any type $\type \geq \monopt$, i.e., 
$\virtual(\type) \geq 0$, 
by applying \cref{lem:bound alloc},
the contribution of revenue from type $\type$ is 
\begin{align}
\Rev(\mech; \type) &\triangleq 
\virtual(\type) \cdot 
\expect[\posterior\sim\alloc(\type)\given\prior]{
\expect[\hat{\posterior}\sim \experiment_{\type,\posterior}\given\posterior]{
\val(\hat{\posterior})}} \nonumber\\
&\leq \virtual(\type) \cdot 
\expect[\posterior\sim\fullsignal\given\prior]{
\expect[\hat{\posterior}\sim \experiment_{\type,\posterior}\given\posterior]{
\val(\hat{\posterior})}} \nonumber\\
&= \virtual(\type) \cdot 
\expect[\posterior\sim\fullsignal\given\prior]{
\val(\posterior)}. \label{eq:rev postive virtual}
\end{align}

We first show that \cref{asp:regular} implies $(\density(\type)\cdot \virtual(\type))' \geq 0$
for any $\type\leq\monopt$. 
If $\density'(\type) \leq 0$, we have that 
\begin{align*}
(\density(\type)\cdot \virtual(\type))' = \density'(\type)\cdot \virtual(\type) + \density(\type)\cdot \virtual'(\type) \geq 0
\end{align*}
since $\virtual(\type)\leq 0$ for any $\type\leq\monopt$.
If $\density'(\type) \geq 0$, 
we have that 
\begin{align*}
(\density(\type)\cdot \virtual(\type))' =
(\density(\type)\cdot\type - (1-\dist(\type)))'
= 2\density(\type) + \density'(\type)\cdot\type\geq 0.
\end{align*}
Next, we bound the revenue contribution from types $\type < \monopt$, i.e., 
$\virtual(\type) < 0$.  
\begin{align}
& \expect[\type\sim\dist]{\Rev(\mech; \type) \cdot \indicate{\type < \monopt}} - \Util(\underline{\type}) \nonumber\\
&= \int_{\underline{\type}}^{\monopt}
\density(\type)\cdot \virtual(\type) \cdot 
\expect[\posterior\sim\alloc(\type)\given\prior]{
\expect[\hat{\posterior}\sim \experiment_{\type,\posterior}\given\posterior]{
\val(\hat{\posterior})}} 
\dd \type
- \Util(\underline{\type}) \nonumber\\
&= - \int_{\underline{\type}}^{\monopt} (\density(\type)\cdot \virtual(\type))'
\int_{\underline{\type}}^{\type}
\expect[\posterior\sim\alloc(z)\given\prior]{\expect[\hat{\posterior}\sim \experiment_{z,\posterior}\given\posterior]{
\val(\hat{\posterior})}} \dd z
\dd \type
- \Util(\underline{\type}) \nonumber\\
&\leq - \int_{\underline{\type}}^{\monopt} (\density(\type)\cdot \virtual(\type))'
\cdot (\V(\prior, \experiments, \type) 
- \Util(\underline{\type})) 
\dd \type
- \Util(\underline{\type}) \nonumber\\
&= - \int_{\underline{\type}}^{\monopt} (\density(\type)\cdot \virtual(\type))'
\cdot \V(\prior, \experiments, \type)
\dd \type
- \Util(\underline{\type}) (\density(\underline{\type})\cdot \virtual(\underline{\type}) + 1) \nonumber\\
&\leq - \int_{\underline{\type}}^{\monopt} (\density(\type)\cdot \virtual(\type))'
\cdot \V(\prior, \experiments, \type)
\dd \type
- \V(\prior, \experiments, \underline{\type}) (\density(\underline{\type})\cdot \virtual(\underline{\type}) + 1) \nonumber\\
&= \int_{\underline{\type}}^{\monopt}
\density(\type)\cdot\virtual(\type) \cdot 
\expect[\posterior\sim \experiment_{\type,\prior}\given\prior]{
\val(\posterior)}
\dd \type
- \V(\prior, \experiments, \underline{\type}).
\label{eq:rev negative virtual}
\end{align}
The second equality holds by integration by parts. 
The first inequality holds by (1) $(\density(\type)\cdot \virtual(\type))'$ is nonnegative for $\type\leq \monopt$,
and (2) $\int_{\underline{\type}}^{\type} 
\expect[\posterior\sim\alloc(z)\given\prior]{
\expect[\hat{\posterior}\sim \experiment_{z,\posterior}\given \posterior]{\val(\hat{\posterior})}} \dd z
\geq \V(\prior, \experiments, \type)
- \Util(\underline{\type})$
according to the individual rationality constraints in \cref{lem:incentive}.
The last inequality holds since 
$\Util(\underline{\type})\geq \V(\prior, \experiments, \underline{\type})$
and $\density(\underline{\type})\cdot \virtual(\underline{\type}) + 1
= \density(\underline{\type}) \cdot \underline{\type} + \dist(\underline{\type}) \geq 0$.
Finally, the last equality holds by integration by parts and the facts that $\virtual(\monopt) = 0$ and
\begin{align*}
\V(\prior, \experiments, \type)
= \V(\prior, \experiments, \underline{\type})
+ \int_{\underline{\type}}^\type \expect[\posterior\sim \experiment_{z,\prior}\given\prior]{
\val(\posterior)}
\dd z.
\end{align*}
Combining \Cref{eq:opt rev,eq:rev postive virtual,eq:rev negative virtual}, we have
\begin{align*}
&\Rev(\mech) \leq 
\expect[\type\sim\dist]{\Rev(\mech; \type) \cdot \indicate{\type \geq \monopt}} 
+ \expect[\type\sim\dist]{\Rev(\mech; \type) \cdot \indicate{\type < \monopt}} - \Util(\underline{\type})\\
&\leq \int_{\monopt}^{\bar{\type}}
\density(\type)\cdot\virtual(\type) \cdot 
\expect[\posterior\sim \fullsignal\given\prior]{
\val(\posterior)}
\dd \type
+\int_{\underline{\type}}^{\monopt}
\density(\type)\cdot\virtual(\type) \cdot 
\expect[\posterior\sim \experiment_{\type,\prior}\given\prior]{
\val(\posterior)}
\dd \type
- \V(\prior, \experiments, \underline{\type}) \\
&= \Rev(\widehat{\mech}). 
\end{align*}
Thus, mechanism $\widehat{\mech}$ is revenue-optimal. 
\end{proof}

\section{Pricing for Full Information}
\label{apx:pricing}
Before the proof of \cref{thm:approx}, we first introduce the definition of quantiles and revenue curves, 
which are helpful for bounding the approximation ratio. 
For any distribution $\dist$, 
let $\quant_{\dist}(\type)\triangleq \Pr_{z\sim\dist}[z\geq\type]$
be the quantile corresponding to type $\type$. 
Accordingly, we can define $\type(\quant)$ as the type corresponds to quantile $\quant$.
The revenue curve as a function of the quantile is defined as $R_{\dist}(\quant) \triangleq \quant\cdot\type(\quant)$.
Note that the regularity condition in \cref{asp:regular}
is equivalent to the concavity assumption for the revenue curve. 
\begin{lemma}[\citealp*{myerson1981optimal}]\label{lem:regular}
A distribution $\dist$ is regular if and only if $R_{\dist}(\quant)$ is concave in~$\quant$.
\end{lemma}

\thmapprox*
\begin{proof}
We first normalize the primitives such that $\monopt\cdot \quant(\monopt) = 1$.
For any type~$\type$, let $\ir(\type)\triangleq \V(\prior, \experiments, \type)$ be the value of the agent
for not participating in the auction. 
It is easy to verify that $\ir(\type)$ is convex in $\type$.
Let $\bx \triangleq \expect[\state\sim\prior]{\val(\state)}$ be the maximum possible allocation. 
According to \cref{thm:optimal linear}, 
if the distribution $\dist$ is regular, 
in the revenue-optimal mechanism, 
the expected utility of the agent is $\ir(\type)$ for any $\type<\monopt$
and is $(\type-\monopt)\cdot \bx + \ir(\monopt)$
for any $\type\geq \monopt$.

Suppose that $\hat{\type}$ is the cutoff type that participates the auction in the optimal price posting mechanism for distribution $\dist$. 
It is easy to verify that $\hat{\type}\leq \monopt$
since revealing full information to any type above the monopoly type only increases the expected revenue. 
Moreover, for any type $\type < \monopt$, 
since the payment that induces $\hat{\type}$ to be the cutoff type
is $\hat{\type}\cdot \bx - \ir(\hat{\type})$,
we have that 
\begin{align*}
(\hat{\type}\cdot \bx - \ir(\hat{\type})) \cdot \quant_{\dist}(\hat{\type})
\geq (\type\cdot \bx - \ir(\type)) \cdot \quant_{\dist}(\type).
\end{align*}
That is, any type $\type < \monopt$, 
\begin{align*}
\quant_{\dist}(\type) \leq 
\frac{(\hat{\type}\cdot \bx - \ir(\hat{\type})) \cdot \quant_{\dist}(\hat{\type})}{\type\cdot \bx - \ir(\type)}.
\end{align*}
\begin{figure}[t]
\centering
\input{fig/eqrev}
\caption{\label{fig:eqrev} 
The figure illustrates the reduction in the type distribution 
that maximizes the approximation ratio between the optimal revenue and the price posting revenue.
The black solid curve is the revenue curve for distribution $\dist$,
and the red dashed curve is the revenue curve for distribution $\hat{\dist}$. 
The black dashed curve is the revenue curve $\bar{\dist}$ such that the seller is indifferent to deterministically selling at any prices with negative virtual value. 
}
\end{figure}
Let $\bar{\dist}$ be the distribution such that 
\begin{align*}
\quant_{\bar{\dist}}(\type) = \frac{(\hat{\type}\cdot \bx - \ir(\hat{\type})) \cdot \quant_{\dist}(\hat{\type})}{\type\cdot \bx - \ir(\type)}
\end{align*}
for any type $\type$.
Thus, the virtual value function $\bar{\virtual}(\type)$ for distribution $\bar{\dist}$
is 
\begin{align*}
\bar{\virtual}(\type) = \type - \frac{\type\cdot \bx - \ir(\type)}{\bx - \ir'(\type)} \leq 0.
\end{align*}
Moreover, 
\begin{align*}
\bar{\virtual}'(\type) = \frac{\ir''(\type) \cdot (\ir(\type) - \type\cdot \bx)}{(\bx - \ir'(\type))^2} 
\leq 0.
\end{align*}
Thus, the revenue curve such that the seller is indifferent to selling at any price is convex. 
Let $\hat{\dist}$ be the distribution with the piecewise linear revenue curve illustrated in \Cref{fig:eqrev}.
Thus, we have that 
\begin{align*}
\quant_{\hat{\dist}}(\type) = \begin{cases}
\frac{1}{\type} & \type\geq \frac{1}{\bar{\quant}},\\
\frac{1-r\bar{\quant}}{\type(1-\bar{\quant})+1-r}
& \type< \frac{1}{\bar{\quant}}.
\end{cases}
\end{align*}
Let $\price(\type)\triangleq \type\cdot\bx - \ir(\type)\geq 0$.
First, note that distribution $\hat{\dist}$ is first-order
stochastically dominated by $\bar{\dist}$, and 
the optimal revenue from posted pricing is weakly smaller for distribution~$\hat{\dist}$.
Moreover, both distributions achieve the same price posting revenue by choosing price $\price(\hat{\type})$ such that the cutoff type is $\hat{\type}$. 
Thus, the optimal price posting revenue for distribution $\hat{\dist}$ is attained by choosing price $\price(\hat{\type})$. 
This further indicates that optimal price posting revenue is the same for distributions $\dist$ and $\hat{\dist}$,
i.e., 
$\pp(\dist,\ir) = \pp(\hat{\dist},\ir)$.
Second, since distribution~$\dist$ is first-order stochastically dominated by $\hat{\dist}$,
it is easy to verify that 
$\opt{\dist,\ir} \leq \opt{\hat{\dist},\ir}$.
Therefore, the ratio between the price posting revenue and the optimal revenue is minimized when the type distribution is $\hat{\dist}$.

For distribution $\hat{\dist}$, since the optimal price is $\price(\hat{\type})$, we have 
\begin{align*}
(\hat{\type}\cdot \bx - \ir(\hat{\type})) \cdot \quant_{\hat{\dist}}(\hat{\type})
\geq (\type\cdot \bx - \ir(\type)) \cdot \quant_{\hat{\dist}}(\type).
\end{align*}
Let $\zeta = (\hat{\type}\cdot \bx - \ir(\hat{\type})) \cdot \quant_{\hat{\dist}}(\hat{\type})$,
and let
\begin{align*}
\hat{\ir}(\type) = \type\cdot \bx - \frac{\zeta}{\quant_{\hat{\dist}}(\type)}.
\end{align*}
It is easy to verify that 
$\pp(\hat{\dist},\ir) = \pp(\hat{\dist},\hat{\ir}) = \zeta$.
Moreover, $\hat{\ir}(\type)$ is convex and $\ir(\type) \geq \hat{\ir}(\type)$ for any $\type$,
which implies that any feasible mechanism for $\ir$ is also feasible for $\hat{\ir}$,
and hence 
$\opt{\hat{\dist},\ir} \leq \opt{\hat{\dist},\hat{\ir}}$.
Thus, to prove \cref{thm:approx}, it is sufficient to bound 
$\frac{\pp(\hat{\dist},\hat{\ir})}{\opt{\hat{\dist},\hat{\ir}}}$.
Note that by construction, the monopoly type for distribution $\hat{\dist}$
is $\frac{1}{\bar{\quant}}$. 
Hence, the optimal revenue is 
\begin{align*}
\opt{\hat{\dist},\hat{\ir}} &= \int_{r}^{\frac{1}{\bar{\quant}}} \hat{\density}(\type) (\type\cdot \hat{\ir}'(\type) - \hat{\ir}(\type)) \dd \type
+ \left(\frac{1}{\bar{\quant}}\cdot \bx - \hat{\ir}(\frac{1}{\bar{\quant}})\right)\cdot \bar{\quant}\\
&= \int_{r}^{\frac{1}{\bar{\quant}}} \hat{\density}(\type) \cdot \zeta\cdot \frac{1-r}{1-r\bar{\quant}} \dd \type
+\zeta\\
&= \zeta\cdot \left(\frac{(1-\bar{\quant})(1-r)}{1-r\bar{\quant}}+1\right) \leq 2\zeta,
\end{align*}
where the inequality is tight if $\bar{\quant} = r = 0$.
Combining the observations, for any distribution $\dist$
and any function $\ir$ induced by the set of experiments $\experiments$, we have 
\begin{equation*}
\frac{\pp(\dist,\ir)}{\opt{\dist,\ir}} 
\geq \frac{\pp(\hat{\dist},\hat{\ir})}{\opt{\hat{\dist},\hat{\ir}}}
\geq \frac{1}{2}.\qedhere
\end{equation*}
\end{proof}

\propinformativeprior*
\begin{proof}
By \cref{thm:optimal linear}, it is sufficient to show that 
if 
$$
\noexp \in \arg\max_{\experiment\in\experiments} \,
\expect[\hat{\posterior} \sim \experiment\given\prior]{\V(\hat{\posterior}, \monopt)}
- \costagent(\experiment,\prior),
$$
then $\experiment_{\type,\prior} = \noexp$
for any $\type < \monopt$.
Suppose by contradiction that there exists 
$\type < \monopt$ such that $\costagent(\experiment_{\type,\prior},\prior)>0$, 
i.e., 
\begin{align*}
\expect[\posterior\sim\experiment_{\type,\prior}\given\prior]{\val(\posterior)}\cdot\type - \costagent(\experiment_{\type,\prior},\prior)
\geq \expect[\posterior\sim\nosignal\given\prior]{\val(\posterior)}\cdot\type. 
\end{align*}
Since $\monopt > \type$, 
we have that 
\begin{align*}
&\expect[\posterior\sim\experiment_{\type,\prior}\given\prior]{\val(\posterior)}\cdot\monopt 
- \costagent(\experiment_{\type,\prior},\prior)
- \expect[\posterior\sim\nosignal\given\prior]{\val(\posterior)}\cdot\monopt \\
>\,& 
\left(\expect[\posterior\sim\experiment_{\type,\prior}\given\prior]{\val(\posterior)}
- \expect[\posterior\sim\nosignal\given\prior]{\val(\posterior)}
\right) \cdot\type 
- \costagent(\experiment_{\type,\prior},\prior)
\geq 0,
\end{align*}
contradicting the assumption that $\noexp$ is one of the optimal choices for type $\monopt$ given prior~$\prior$.
\end{proof}

In the following lemma, 
in the case in which the state space $\states$ is finite,
we formalize the intuition that when the prior is sufficiently close to the degenerate pointmass distribution,
the marginal cost of additional information exceeds the marginal benefit of additional information.
\begin{definition}
The set of possible experiments $\experiments$ 
\emph{is finitely generated}
if it is generated by $\noexp$ and a finite set~$\experiments'$ through sequential learning, 
where $\noexp$ is the experiment that always reveals no additional information with zero cost, 
and any $\experiment \in |\experiments'|$ is an experiment that provides an informative signal about the state with fixed cost $c_{\experiment}>0$.
\end{definition}
\begin{restatable}{lemma}{leminformativeprior}\label{lem:informative prior}
Suppose that $\states$ is finite and $\experiments$ is finitely generated.
Suppose that there exists $\bar{\val} < \infty$ such that $\max_{\state\in\states}\val(\state) \leq \bar{\val}$
and $\val(\posterior) \geq \min_{\state\in\states} \posterior(\state) \cdot\val(\state)$ for any $\posterior$.
Then, there exists $\epsilon > 0$
such that any prior $\prior$
satisfying $\prior(\state) > 1-\epsilon$ for some $\state \in \states$
is sufficiently informative.\footnote{We can have similar results when $\experiments$ is not finitely generated. 
For example,
when the cost function is the reduction in entropy,
by applying the techniques in \citet*{caplin2019rational}, 
for any valuation function $\val$,
there exists $\epsilon>0$ such that
any prior $\prior$
satisfying $\prior(\state) > 1-\epsilon$ for some $\state \in \states$
is sufficiently informative.
} 
\end{restatable}
\begin{proof}
Let $c_m = \min_{\experiment\in\experiments'} > 0$.
By construction, for any experiment $\experiment \in \experiments$, we have $\costagent(\experiment,\posterior) \geq c_m$ for any $\posterior$.
Let $\state^*$ be the state such that 
$\prior(\state^*) > 1-\epsilon$.
Given prior $\prior$, the utility increase of type $\monopt$ from additional information 
is at most 
\begin{align*}
\monopt \cdot\left(\sum_{\state\in\states} \prior(\state) \val(\state) - \val(\prior)\right)
\leq \monopt \cdot\left(\sum_{\state\neq\state^*} \prior(\state) \val(\state)\right)
< \monopt \cdot \epsilon \cdot \bar{\val}.
\end{align*}
The first inequality holds since $\val(\prior) \geq \prior(\state^*) \val(\state^*)$, 
and the second inequality holds since 
$\val(\state) \leq \bar{\val}$
and $\sum_{\state\neq\state^*} \prior(\state) < \epsilon$.
Thus, when $\epsilon = \frac{c_m}{\monopt \cdot \bar{\val}}$, 
the cost of information is always higher than the benefit of information, 
and the agent with type $\monopt$ will never acquire any additional information given prior $\prior$.
\end{proof}

\section{Monopoly Auction}
\label{apx:monopoly_exp}
In this section, we examine \cref{exp:firm auction} 
and show that an agent who
is a firm selling a product with different quality levels 
satisfies the linear valuation assumption. 
In fact, it is sufficient to show that for any $\posterior$ 
and any pair of cost functions $c,\hat{c} > 0$, 
letting $\type = \frac{1}{c},\hat{\type}=\frac{1}{\hat{c}}$,
we have $\frac{\V(\posterior,\hat{\type})}{\V(\posterior,\type)} = \frac{\hat{\type}}{\type}$.

Given seller type $\type$, for any mechanism with allocation and payment rule $\alloc,\price$,
consider an alternative mechanism with an allocation and payment rule $\hat{\alloc},\hat{\price}$,
such that 
\begin{align*}
\hat{\alloc}(\type) = \alloc(\type)\cdot\frac{\hat{\type}}{\type}
\quad\text{and}\quad
\hat{\price}(\type) = \price(\type)\cdot\frac{\hat{\type}}{\type},
\quad\forall \type.
\end{align*}
It is easy to verify that $(\hat{\alloc},\hat{\price})$ satisfies incentive compatibility and individual rationality. 
The expected revenue given $(\hat{\alloc},\hat{\price})$ when the cost is $\hat{c}$ is 
\begin{align*}
\expect[\type]{\hat{\price}(\type) - \frac{1}{2\hat{\type}}\cdot \hat{\alloc}^2(\type)}
= \expect[\type]{\price(\type)\cdot\frac{\hat{\type}}{\type} - \frac{1}{2\hat{\type}}\cdot \rbr{\alloc(\type)\cdot\frac{\hat{\type}}{\type}}^2}
= \frac{\hat{\type}}{\type}\cdot\expect[\type]{\price(\type) - \frac{1}{2\type}\cdot \alloc(\type)}.
\end{align*}
Since the above argument holds for any mechanism $(\alloc,\price)$,
the optimal revenue when the agent's type is $\hat{\type}$ should be at least 
$\frac{\hat{\type}}{\type}$ times the optimal revenue when the agent's type is $\type$,
i.e., $\frac{\V(\posterior,\hat{\type})}{\V(\posterior,\type)} \geq \frac{\hat{\type}}{\type}$.
By applying the same argument when reversing the role of $\type$ and $\hat{\type}$, 
we also have that $\frac{\V(\posterior,\hat{\type})}{\V(\posterior,\type)} \leq \frac{\hat{\type}}{\type}$.
Combining the inequalities, we have $\frac{\V(\posterior,\hat{\type})}{\V(\posterior,\type)} = \frac{\hat{\type}}{\type}$, 
and the linearity assumption is satisfied.

\end{appendix}

\end{document}

%% file: notation.tex
\newcommand{\cost}{C}
\newcommand{\costagent}{C^A}
\newcommand{\alloc}{\pi}
\newcommand{\price}{p}

\newcommand{\Util}{U}

\newcommand{\Rev}{{\rm Rev}}
\newcommand{\opt}[1]{{\rm OPT}(#1)}
\newcommand{\virtual}{\phi}

\newcommand{\quant}{q}

\newcommand{\var}{{\rm Var}}
\newcommand{\type}{\theta}
\newcommand{\types}{\Theta}
\newcommand{\monopt}{\type^*}
\newcommand{\experiment}{\hat{\sigma}}
\newcommand{\experiments}{\hat{\Sigma}}

\newcommand{\prior}{D}
\newcommand{\posterior}{\mu}
\newcommand{\state}{\omega}
\newcommand{\states}{\Omega}
\newcommand{\V}{V}
\newcommand{\val}{v}
\newcommand{\signal}{\sigma}
\newcommand{\signals}{\Sigma}
\newcommand{\signalspace}{S}
\newcommand{\signalspacea}{S}
\newcommand{\mech}{\mathcal{M}}
\newcommand{\s}{s}
\newcommand{\sa}{s}
\newcommand{\entropy}{H}
\newcommand{\actions}{A}
\newcommand{\action}{a}
\newcommand{\dist}{F}
\newcommand{\density}{f}
\newcommand{\full}{^F}
\newcommand{\no}{^N}
\newcommand{\fullsignal}{\signal\full}
\newcommand{\nosignal}{\signal\no}
\newcommand{\fullexp}{\signal\full}
\newcommand{\noexp}{\signal\no}
\newcommand{\g}{\mathcal{N}}
\newcommand{\ir}{c}
\newcommand{\x}{x}
\newcommand{\bx}{\bar{\x}}
\newcommand{\pp}{{\rm PP}}

%% file: math.tex
\DeclareMathOperator{\argmax}{argmax}

\newcommand{\given}{|}

\newcommand{\prob}[2][]{\text{\bf Pr}\ifthenelse{\not\equal{}{#1}}{_{#1}}{}\!\left[{\def\givenn{\middle|}#2}\right]}
\newcommand{\expect}[2][]{\text{\bf E}\ifthenelse{\not\equal{}{#1}}{_{#1}}{}\!\left[{\def\givenn{\middle|}#2}\right]}

\newcommand{\tparen}{\big}
\newcommand{\tprob}[2][]{\text{\bf Pr}\ifthenelse{\not\equal{}{#1}}{_{#1}}{}\tparen[{\def\given{\tparen|}#2}\tparen]}
\newcommand{\texpect}[2][]{\text{\bf E}\ifthenelse{\not\equal{}{#1}}{_{#1}}{}\tparen[{\def\given{\tparen|}#2}\tparen]}

\newcommand{\sprob}[2][]{\text{\bf Pr}\ifthenelse{\not\equal{}{#1}}{_{#1}}{}[#2]}
\newcommand{\sexpect}[2][]{\text{\bf E}\ifthenelse{\not\equal{}{#1}}{_{#1}}{}[#2]}

\newcommand{\rbr}[1]{\left(#1\right)}

\newcommand{\dd}{{\,\mathrm d}}

\newcommand{\indicate}[1]{{\bf 1}\left[#1\right]}
\newcommand{\reals}{\mathbb{R}}

%% file: fig/auction.tex
\begin{tikzpicture}[scale = 0.48]

\draw (-0.2,0) -- (12.5, 0);
\draw (0, -0.2) -- (0, 7.2);


\begin{scope}

\draw  plot [smooth, tension=0.7] coordinates { (0,2.5) (1.6, 6.5) (6, 4)};

\draw  plot [smooth, tension=0.7] coordinates { (6, 4) (10, 5.5) (12, 0.5)};
\end{scope}

\draw (0, -0.8) node {$0$};

\draw (12, 0) -- (12, 0.2);
\draw (12, -0.8) node {$1$};

\draw (13.6, 0) node {$\hat{\posterior}(\state_1)$};


\draw (1.4, 8) node {$\val(\hat{\posterior}) \cdot \type + \entropy(\hat{\posterior})$};

\draw [dotted] (12, 0) -- (12, 0.5);

\draw (1.9, -0.8) node {$\posterior_{\type}(\state_1)$};
\draw (1.9, 0) -- (1.9, 0.2);
\draw [dotted] (1.9, 0) -- (1.9, 6.5);
\draw (9.7, -0.8) node {$\posterior'_{\type}(\state_1)$};
\draw (9.7, 0) -- (9.7, 0.2);
\draw [dotted] (9.7, 0) -- (9.7, 5.7);


\draw [dashed] (0.8, 6.67) -- (12, 5.38);

\end{tikzpicture}

%% file: fig/eqrev.tex
\begin{tikzpicture}[scale = 0.48]

\draw (-0.2,0) -- (12.5, 0);
\draw (0, -0.2) -- (0, 7.2);

\begin{scope}

\draw plot [smooth, tension=0.6] coordinates { (0,0) (2,5) (3, 6)};
\draw plot [smooth, tension=0.7] coordinates { (3, 6) (9, 5) (12, 0)};

\draw [dashed] plot [smooth, tension=0.7] coordinates { (8, 7) (9, 5) (12, 4)};
\end{scope}

\draw [red, dashed, thick] (0, 6) -- (7.5, 6);
\draw [red, dashed, thick] (7.5, 6) -- (12, 3);
\draw [red, dashed, thick] (12, 3) -- (12, 0);

\draw (-0.5, 6) node {$1$};
\draw (-0.5, 3) node {$r$};
\draw [dotted] (0, 3) -- (12, 3);

\draw (0, -0.8) node {$0$};
\draw (12, 0) -- (12, 0.2);
\draw (12, -0.8) node {$1$};

\draw (13.6, 0) node {$\quant$};

\draw (3, 0) -- (3, 0.2);
\draw [dotted] (3, 0) -- (3, 6);
\draw (3, -0.8) node {$q(\monopt)$};

\draw (9, 0) -- (9, 0.2);
\draw [dotted] (9, 0) -- (9, 5);
\draw (9, -0.8) node {$q(\hat{\type})$};

\draw (7.5, 0) -- (7.5, 0.2);
\draw [dotted] (7.5, 0) -- (7.5, 6);
\draw (7.5, -0.8) node {$\bar{q}$};

\draw (0, 8) node {$R(\quant)$};





\end{tikzpicture}